\documentclass[article]{imsart}
\usepackage{grffile}

\usepackage{amsmath,amsthm,amsfonts,amssymb,amscd}
\RequirePackage[numbers]{natbib}
\RequirePackage[colorlinks,citecolor=blue,urlcolor=blue]{hyperref}
\RequirePackage{graphicx}
\usepackage{subfig}
\usepackage{comment}
\usepackage{natbib}
\usepackage[margin=1in]{geometry}
\usepackage{mathtools}
\newcommand{\floor}[1]{\lfloor #1 \rfloor}
\usepackage{enumerate}
\usepackage{graphicx}
\usepackage{listings}
\usepackage{float}
\allowdisplaybreaks
\usepackage{natbib}
\usepackage{subfig}

\usepackage{lipsum}
\usepackage{epsfig}
\usepackage{bm}
\usepackage{multicol}
\usepackage{url}
\usepackage[utf8]{inputenc}
\usepackage{mathtools}
\usepackage{extarrows}

\usepackage{color}
\definecolor{dkgreen}{rgb}{0,0.6,0}
\definecolor{gray}{rgb}{0.5,0.5,0.5}
\definecolor{mauve}{rgb}{0.58,0,0.82}

\usepackage{booktabs}
\usepackage{longtable}
\usepackage{array}
\usepackage{multirow}
\usepackage{wrapfig}
\usepackage{colortbl}
\usepackage{pdflscape}
\usepackage{tabu}
\usepackage{threeparttable}
\usepackage{threeparttablex}
\usepackage[normalem]{ulem}
\usepackage{makecell}
\usepackage{xcolor}

\usepackage[section]{placeins}

\newcommand{\blue}[1]{{\leavevmode\color{black}{#1}}}

\newcommand{\wzi}{\widehat{Z_i}}

\newcommand{\abs}[1]{\left\lvert#1\right\rvert}
\newcommand{\norm}[1]{\left\lVert#1\right\rVert}

\newcommand{\calB}{\mathcal B}
\newcommand{\calX}{\mathcal X}
\newcommand{\calZ}{\mathcal Z}

\newcommand{\calS}{\mathcal S}

\newcommand{\calF}{\mathcal F}

\newcommand{\bI}{ {\boldsymbol I} }

\newcommand{\bk}{ {\boldsymbol k} }
\newcommand{\bK}{ {\boldsymbol K} }

\newcommand{\bM}{ {\boldsymbol M} }

\newcommand{\bs}{ {\boldsymbol s} }

\newcommand{\bt}{ {\boldsymbol t} }

\newcommand{\bx}{ {\boldsymbol x} }
\newcommand{\bX}{ {\boldsymbol X} }

\newcommand{\bY}{ {\boldsymbol Y} }

\newcommand{\eps}{\epsilon}

\newcommand{\given}{\,|\,}

\newcommand{\bgamma}{ {\boldsymbol \gamma} }

\newcommand{\bmu}{ {\boldsymbol \mu} }

\newcommand{\bSigma}{ {\boldsymbol \Sigma} }

\newcommand{\iid}{\overset{\mbox{i.i.d.}} \sim}
\newcommand{\ind}{\overset{\mbox{ind}} \sim}

\newcommand{\logit}{\text{logit}}

\newtheorem{theorem}{Theorem}
\newtheorem{definition}{Definition}
\newtheorem{lemma}{Lemma}

\usepackage{xr}

\makeatletter
\newcommand*{\addFileDependency}[1]{
	\typeout{(#1)}
	\@addtofilelist{#1}
	\IfFileExists{#1}{}{\typeout{No file #1.}}
}\makeatother

\pubyear{2023}
\volume{0}
\issue{0}
\firstpage{0}
\lastpage{0}

\startlocaldefs
\numberwithin{equation}{section}
\theoremstyle{plain}

\endlocaldefs

\makeatletter
\newcommand\footnoteref[1]{\protected@xdef\@thefnmark{\ref{#1}}\@footnotemark}
\makeatother

\begin{document}

\begin{frontmatter}
\title{Direct Bayesian \blue{Linear} Regression for Distribution-valued Covariates
}
\runtitle{Bayesian Distribution Regression}

\begin{aug}
\author{\fnms{Bohao} \snm{Tang}
\ead[label=e1]{bhtang127@gmail.com}}\footnote{\label{note1} These authors contributed equally to the manuscript.}

\address{Department of Biostatistics\\ Johns Hopkins University\\
\printead{e1}}

\author{\fnms{Sandipan} \snm{Pramanik}
\ead[label=e2]{spraman4@jhmi.edu}}\blue{$^*$}

\address{Department of Biostatistics\\ Johns Hopkins University\\
\printead{e2}}

\author{\fnms{Yi} \snm{Zhao}
\ead[label=e3]{zhaoyi1026@gmail.com}}

\address{Department of  Biostatistics \& Health Data Science\\ Indiana University\\
\printead{e3}}

\author{\fnms{Brian} \snm{Caffo}
\ead[label=e4]{bcaffo@gmail.com}}

\address{Department of Biostatistics\\ Johns Hopkins University\\
\printead{e4}}

\author{\fnms{Abhirup} \snm{Datta}
\ead[label=e5]{abhidatta@jhu.edu}}

\address{Department of Biostatistics\\ Johns Hopkins University\\
\printead{e5}}

\runauthor{B. Tang and S. Pramanik et al.}

\end{aug}

\begin{abstract}
 In this manuscript, we study scalar-on-distribution regression; that is, instances where  subject-specific distributions or densities are the covariates related to a scalar outcome via a regression model. In practice, only repeated measures are observed from those covariate distributions and common approaches first use these to estimate subject-specific density functions, which are then used as covariates in standard scalar-on-function regression. We propose a simple and direct method for \blue{linear} scalar-on-distribution regression that circumvents the intermediate step of estimating subject-specific covariate densities. We show that one can directly use the observed repeated measures as covariates and endow the regression function with a Gaussian process prior to obtain closed form or conjugate Bayesian inference. Our method subsumes the standard Bayesian non-parametric regression using Gaussian processes as a special case, \blue{corresponding to covariates being Dirac-distributions}.
 The model is also invariant to any transformation or ordering of the repeated measures. Theoretically, we show that, despite only using the observed repeated measures from the true density-valued covariate that generated the data, the method can achieve an optimal estimation error bound of the regression function. \blue{The theory extends beyond i.i.d. settings 
 to accommodate certain forms of within-subject dependence 
 among the repeated measures.}
 To our knowledge, this is the first theoretical study on Bayesian  regression using distribution-valued covariates. \blue{We propose numerous extensions including a scalable implementation using low-rank Gaussian processes and a generalization to non-linear scalar-on-distribution regression.} Through simulation studies, we demonstrate that our method performs \blue{substantially} better than approaches that require an intermediate density estimation step \blue{especially with a small number of repeated measures per subject}. \blue{We apply our method to study association of 
 age with activity counts.} 
\end{abstract}


\tableofcontents
\end{frontmatter}

\section{Introduction}
\label{sec:intro}

Regression with distribution- or density-valued objects as outcomes, covariates, or both, is called ``{\it distribution regression}''~\citep{petersen2016functional,szabo2016learning, oliva2014fast, fang2020optimal, law2018bayesian,chen2021wasserstein}. 
This manuscript considers regression with a scalar outcome and a distribution-valued covariate. Applications of such {\it scalar-on-distribution} regression include studying association of a subject's health endpoints with the entire distribution of air pollution concentrations that the subject is exposed to, or using posterior distributions of variables  from a Bayesian analysis as covariates in a secondary analysis.

In practice, no collected data is truly density-valued, and one only observes a finite set of repeated \blue{measures or} samples from the covariate distribution (e.g., repeated measurements of air pollution concentrations from sensors, or a collection of MCMC samples as estimates of true posterior distributions). 
That is, the observed covariate of subject $i$ represents a conceptual two-stage sampling procedure, first of a subject-specific distribution \blue{$Z_i$} from a distribution on distributions, and secondly of samples (repeated measures) $\{\bx_{ij}\}$ from the resulting subject-specific distribution $Z_i$.  
Even though the outcome $y_i$ depends on the entire distribution $Z_i$, only the samples $\{\bx_{ij}\}$ are observed. 
 A common practice for the analysis of distribution-valued covariates is thus to estimate 
 \blue{a density ${\widehat{dZ_i}}$ of the distribution $Z_i$ from the samples $\{\bx_{ij}\}$ using kernel density estimation (KDE)} and to utilize the classic scalar-on-functional regression model \cite{ramsilver} \blue{with the density functions $\widehat{dZ_i}$ as covariates, i.e.,}  
 $\mathbb E_{\varepsilon_i} y_i = \int f(\bx) \blue{\widehat{dZ_i}(\bx)}d\bx$ to estimate the unknown non-linear regression function $f$ \citep{talska2021compositional}.
However, kernel methods often do not perform well in density estimation due to boundary effects~\citep{petersen2016functional}. \blue{We will also see they perform poorly when the number of repeated measures per subject is small or if the true covariate distributions are highly discrete.} \citet{augustin2017modelling} adopted a similar approach for accelerometer readings by summarizing the repeated measures into a histogram, i.e., a vector of relative frequencies for each bin, which is then treated as the covariates in a linear regression model. Theoretical properties of the method were not explored. Additionally, histograms introduce unnecessary discretization and  
any density or histogram based approach to distribution regression involve additional tuning (bandwidth) parameter selection steps. 
\blue{These approaches can be thought of as {\em linear} scalar-on-distribution regression as the analysis model used is simply approximating a true data generation process where the mean of the response $\mathbb E_{\varepsilon_i} y_i = \int f(\bx) Z_i(d\bx)$ is a linear functional of the distribution $Z_i$, although the regression coefficient $f$ specifying the functional can be non-linear.} 
\citet{poczos2013distribution,oliva2014fast} developed \blue{non-linear scalar-on-distribution regression where the kernel density estimates are related to the scalar outcome by a second kernel with density-valued inputs. These methods, still relying on a first-stage density estimation,  inherits the aforementioned issues of KDE.} 

Another direction of performing distribution regression, \blue{that circumvents density estimation} is to embed the distributions into a reproducing kernel Hilbert space (RKHS) via a kernel mean embedding and then apply an RKHS regression. 
The method can \blue{also} be regarded as a fully non-linear distribution regression, deploying two kernels -- one mapping distributions to mean embeddings and the second mapping mean embeddings to the response space. 
\citet{szabo2016learning} showed that such methods can achieve minimax error bounds.  
However, the theory is primarily concerned with prediction error bounds and
relies on a restrictive assumption of bounded outcomes. 

In this paper, we propose a simple and direct \blue{Bayesian} scalar-on-distribution regression approach that circumvents the intermediate step of estimating covariate densities. 
We \blue{primarily focus on the linear scalar-on-distribution regression setting and} adopt a Bayesian non-parametric approach to  \blue{modeling the unknown regression function $f$}, using a Gaussian process (GP) prior. We propose modeling the mean $\mathbb E_{\varepsilon_i} y_i$  simply as a Lebesgue integral of the GP $f(x)$ over the empirical distribution \blue{$\wzi$} of the observed samples $\{\bx_{ij}\}$. The resulting model is also a GP regression, with an {\em `average kernel'}, averaging over the observed samples. We thus \blue{proffer a Bayesian GP regression} for distribution-valued covariates that subsumes the standard \blue{non-parametric} GP regression with scalar- or vector-valued covariates as a special case \blue{corresponding to} Dirac (point-mass) distributions.  
We show that \blue{our proposed} GP distribution regression 
retains all advantages of Gaussian process, including 
flexibility of use in hierarchical Bayesian models, closed-form or conjugate posterior inference, 
and amenability to scalable implementation using GP approximations.
Our approach is invariant \blue{to ordering} and 
transformations of the observed samples.

A central contribution of the manuscript is studying properties of the proposed \blue{Bayesian} Gaussian process regression for distribution-valued covariates. For scalar- or vector-valued covariates, asymptotics for GP regression have been well established \citep{van2008rates, van2011information, sniekers2015adaptive, choi2007alternative}. 
Notably, one can bound not only the prediction error but also the estimation error, $\|\hat{f} - f_0\|$, where $\hat{f}$ is the estimate of the true regression function $f_{0}$. We establish minimax information rates of our  GP regression for distribution-valued covariates. Unlike the standard theory of GP regression where the analysis model is correctly specified, our theory throughout accounts for the fact that our \blue{analysis model} is misspecified as it only uses the observed samples $\{\bx_{ij}\}$ instead of the true distributions $Z_i$ generating the data. Our theory does not rely on any boundedness assumption of the outcome. We focus more on the estimation of the regression function, $f_0$, and quantify uncertainties in Bayesian modeling via bounding the posterior risk. Under practically acceptable assumptions on the distributional process of $Z$, we obtain optimal estimation error bounds of the regression function $f_0$. \blue{We also extend the theory beyond the setting of i.i.d. repeated measures per subject to accommodate certain forms of within-subject dependence (geometrically absolutely-regular mixing) among the repeated measures, as in practice, repeated measures will often exhibit correlation based on time of measurements.}

To our knowledge, this is the first theoretical work on information rates of Bayesian distribution regression. In fact, \citet{law2018bayesian} is the only Bayesian distribution regression method we observed in the literature. That approach did not directly use the observed samples, but rather used a GP prior on the kernel mean embeddings to estimate the subject-specific distributions. It also used landmark points rather than all the samples 
and considered a restrictive set of regression functions. More importantly, no theoretical properties were established. We instead consider a very flexible class of regression functions and model them in a non-parametric fashion using GP. Our approach is more parsimonious by directly using the observed samples and circumventing estimation of the subject-specific distributions (via their densities or mean embeddings). 

\blue{We propose numerous extensions of the approach including a scalable implementation that exploits low-rank approximations of GP, an extension to fully non-linear scalar-on-distribution regression, accommodating additional (vector-valued) covariates or clustered/grouped data.} 
Via numerical studies, we show that our approach outperforms alternatives that use this intermediate density estimation step, 
\blue{In particular, we see substantial improvement in accuracy from using our method (often with $100$-fold reduction in estimation error) when the number of repeated measures per subject is small. An analysis of \blue{age and} activity count dataset demonstrates that distribution regression performs better than functional regression implying that time-invariant summaries (distributions) of activity are more more associated with age than temporal patterns of activity.}


\section{Method}\label{sec:methods}

\subsection{\blue{Review of Gaussian process regression}}\label{sec:gp}

We first briefly review the standard Bayesian non-parametric regression for scalar- or vector-valued covariates using Gaussian processes. \blue{Note that throughout the manuscript we will clearly distinguish between the true {\em generative model} which generated the data and an {\em analysis model} used to estimate the regression function. While this distinction is not critical for the setting of GP regression with scalar- or vector-valued covariates, it has implications when eventually considering distribution-valued covariates as the generative model will involve the unobserved true distribution $Z_i$ while any analysis model can only work with the observed repeated measures $\bx_{ij}.$}

\blue{Consider a generative} non-linear regression model given by
\begin{equation}\label{eq:gpreg}
y_i = f_0(\bx_i) + \varepsilon_i, \varepsilon_i \iid N(0,\sigma^2)
\end{equation}
where $i$ is the subject index, $\bx_i \blue{\in \mathbb R^d}$, for some \blue{positive integer} $d$, is the scalar- or vector-valued covariate,  $f_0$ is the unknown smooth regression function linking $\bx_i$ to the outcome $y_i$, and $\epsilon_i$ are random errors. 

Gaussian process regression is a Bayesian non-parametric method, widely employed in functional or spatial settings. \blue{For a generative model (\ref{eq:gpreg}), GP regression considers an analysis model 
\begin{equation}\label{eq:gpan}
y_i = f(\bx_i) + \varepsilon_i,\, f \sim GP(0,K(\cdot,\cdot)), \, \varepsilon_i \iid N(0,\sigma^2), 
\end{equation}
i.e.,} it assigns a zero-mean GP prior to the function $f$  
where $K$ is the covariance kernel operator \citep{williams2006gaussian}. This implies a multivariate Gaussian prior for any finite collection of $\bx_i$ such that $Cov(f(\bx_i),f(\bx_j))=K(\bx_i,\bx_j)$.
Gaussian process has a close connection to the RKHS regression in that the posterior mean of GP regression is the same as the minimizer of an RKHS regression under the same kernel and suitable regularization parameters \citep{kanagawa}. Given data from $n$ subjects, with covariates $\bx_1,\ldots,\bx_n$ and assuming i.i.d. Gaussian errors $\varepsilon_i$ with variance $\sigma^2$ and the kernel parameters to be known, \blue{for any $\bs \in \mathbb R^d$,} the posterior distribution of $f$ is available in closed form as 
\begin{align}\label{eq:krig}
	\begin{split}
		\mathbb{E}[f(\bs)\ |\ X, Y] &= \bm{k}^T(s)(K + \sigma^2)^{-1} Y, \\
		\text{Cov}[f(\bs)\ |\ X, Y] &= K(\bs, \bs) - \bm{k}^T(\bs) (\bK + \sigma^2)^{-1} \bm{k}(s) 
	\end{split}
\end{align}
where \blue{$Y=(y_1,\ldots,y_n)$, $X=(\bx_1,\ldots,\bx_n)$}, $\bk^\top(\bs)=(K(\bs,\bx_1),\ldots,K(\bs,\bx_n))$ and $\bK=(K(\bx_i,\bx_j))$. \blue{Note the expectation and covariance in (\ref{eq:krig}) is with respect to the posterior probability distribution, conditional on the data, implied by the Bayesian analysis model (\ref{eq:gpan}), not with respect to multiple draws of the generative model (\ref{eq:gpreg}).} This ability to yield conjugate inference on an unknown function is a central reason for the popularity of GP in spatial, spatio-temporal or functional applications. In practice, as the variance and kernel parameters are unknown, inference is obtained using Markov chain Monte Carlo (MCMC) techniques with the steps in (\ref{eq:krig}) used as conditional updates within the sampler.

\subsection{Gaussian process regression \blue{with distribution-valued covariates}}\label{sec:dgp}



We now propose a simple extension of GP regression to distribution-valued covariates. \blue{We first consider the case where the true $Z_i$ are observed instead of the repeated measures $\bx_{ij}$. We extend the generative model (\ref{eq:gpreg}) for this distributional setting. In (\ref{eq:gpreg}), $y_i$ is associated with a single covariate value $\bx_i$, with the effect quantified as $f_0(\bx_i)$. In a distributional setting, the response is not associated with a single value of the covariate but the entire distribution of covariate values. So the effect size can be expressed as the weighted average of $f_0(\bx)$ with weights proportional to the frequency of the covariate taking value $\bx$. Formally, if $\bx_i$ denote a random draw from $Z_i$, then the effect will be $\int f_0(\bx_i) Z_i(d\bx_i)$, yielding the following generative model for distributional regression:}
\begin{equation}\label{eq:dgp}
\begin{split}
y_i &= \mathbb{E}_{Z_i} f_0 + \varepsilon_i, , \varepsilon_i \iid N(0,\sigma^2), \mbox{ where }\\
\blue{ \mathbb{E}_{Z_i} f_0}&\blue{:=\mathbb{E}_{\bx_i \sim Z_i} f_0(\bx_i) = \int f_0(\bx_i) Z_i(d\bx_i).}
\end{split}
\end{equation}
\blue{The generative model in (\ref{eq:dgp}) is the typical linear scalar-on-distribution regression considered in the literature \cite{poczos2013distribution,oliva2014fast}.} We have thus showed that distribution-regression can be represented as a Gaussian process regression with  latent random inputs $\bm \bx_i$ (which can be viewed as random effects), endowed with the subject-specific distributions $Z_i$. 
\blue{For data generated as (\ref{eq:dgp}), if the $Z_i$ were known, then a reasonable analysis model would be the following extension of the GP regression (\ref{eq:gpan}):}
\begin{align}\label{eq:gpan2}
\begin{aligned}
y_i &= \mathbb{E}_{Z_i} f + \varepsilon_i, \, \blue{f \sim GP(0,K(\cdot,\cdot)), \, \varepsilon_i \iid N(0,\sigma^2).} 
\end{aligned}
\end{align}

\blue{Let $\calS$ denote a collection of distributions from which the subject-specific distributions $Z_i$ are drawn and $\mathcal X$ denote a compact subset of $\mathbb R^d$ containing the supports of each $Z_i$.} When $f$ is assigned a Gaussian process prior on \blue{$\mathcal X$}, 
under regularity conditions ensuring that the integrals $\blue{\mathbb{E}_{Z_i} f=}\int f(\bx_i)(\omega) Z_i(d\bx_i)$ are well-defined on almost all sample paths $f(\cdot)(\omega)$,  the resulting marginal mean from (\ref{eq:gpan2}) is a Gaussian process \blue{on $\calS$.} Formally, if $f \sim GP(0,K(\cdot,\cdot))$ \blue{on $\mathcal X$}  then 
\begin{equation}\label{eq:meta}
\begin{array}{c}
     F(Z) = \mathbb E_{Z}(f) \sim GP(0,\mathbb K(\cdot,\cdot)), 
     \mbox{ where }\\
     \mathbb K(Z,Z') = \int_{\bx} \int_{\bx'} K(\bx,\bx') dZ(\bx) dZ'(\bx'),\, \forall Z,Z' \in \blue{\mathcal S.}
\end{array}
\end{equation}

It is now evident, that 
(\ref{eq:gpan2}) is a direct generalization of the standard Gaussian process regression to distribution-valued covariates. The induced kernel $\mathbb K$ on the distribution space is simply the double expectation of the original kernel $K$ over the pair of distributions. \blue{We refer to (\ref{eq:gpan2}) as the 
{\em oracle} GP regression model as it relies on the knowledge of the true distributions $Z_i$.}

If $Z_i$ is a Dirac distribution with all mass at $\bx_i$, 
then $\mathbb{E}_{\bx_i \sim Z_i}(f) = f(\bm \bx_i)$ and \blue{the generative model (\ref{eq:dgp}) reduces exactly to the generative model (\ref{eq:gpreg}) and the analysis model} (\ref{eq:gpan2}) reduces exactly to \eqref{eq:gpan}. Hence, the distributional formulation is truly a generalization of the \blue{non-linear scalar-on-vector regression.}  


If the distributions $Z_i$ or the corresponding density functions were observed, one can first obtain closed-form updates for the posterior distribution of the Gaussian process $F(Z)$ for any distribution $Z \in \blue{\mathcal S}$, akin to (\ref{eq:krig}) but with the kernel $\mathbb K$. Subsequently, the posterior distribution of the function $f(\bs)$ for any  $s \in \mathbb R^d$ can be obtained by setting $Z = \delta_\bs$ the Dirac distribution at $\bs$, as $f(\bs)=F(\delta_\bs)$.

\subsection{Direct approach \blue{with observed samples}}\label{sec:an}
In practice, \blue{the true distributions} $Z_i$ or the \blue{corresponding} densities are never directly observed as no data is truly density-valued. Instead, we have $m_i$ repeated measures of the covariate, i.e., for each subject, we have $m_i$ i.i.d samples from $Z_i$ denoted by  $\{\bx_{ij}\}_{j=1}^{m_i}$, with $\bx_{ij} \in \mathbb{R}^d$. \blue{Note that none of the 
$\bx_{ij}$ individually correspond to $y_i$, they simply represent draws from the distribution $Z_i$ that is related to $y_i$ as in (\ref{eq:dgp}).} So, the \blue{full} generative model can be summarized as 
\begin{equation}\label{eq:model}
y_i = \mathbb{E}_{Z_i} f_0 + \varepsilon_i,\, \varepsilon_i \iid N(0,\sigma^2),\, \{\bx_{ij} \}_j \iid Z_i, \, Z_i \iid  \mathcal{Z}.
\end{equation}
\blue{Here $\mathcal{Z}$ denotes a distributional process (distribution on the space of distributions) with support $\calS$ from which each $Z_i$ is drawn.}

We now present a direct approach \blue{for estimation of $f_0$ for data generated from (\ref{eq:model})} that circumvents intermediate \blue{density estimation} for the distributions $Z_i$. 
\blue{Let $\wzi$ denote the empirical distribution of the samples. 
We propose the following analysis model by simply replacing $Z_i$ in (\ref{eq:gpan2}) by $\wzi$, i.e.,}
\begin{equation} \label{gp_dr}
\hspace{-0.5em} y_i = \mathbb{E}_{\wzi} f + \varepsilon_i, \blue{ = \frac 1{m_i} \sum_{j=1}^{m_i} f(\bx_{ij}) + \varepsilon_i,} \, \blue{f \sim GP(0,K(\cdot,\cdot)), \, \varepsilon_i \iid N(0,\sigma^2).} 
\end{equation}


We emphasize that the model (\ref{gp_dr}), \blue{is an analysis model} and not a generative model. The repeated measures, $\{\bx_{ij}\}_j$, are random samples from the true subject-specific distributions, and different sets of random samples for the same subject would then change the data generating model for the same observed outcome. It is more likely that the outcome is generated based on the underlying distribution from which the covariate samples are generated, and the generative model should be agnostic to the choice of samples used for the analysis. The generative model \blue{specified by (\ref{eq:model}) ensures this.} To elucidate with an example, a subject's cardio-vascular or pulmonary health end-point  will be more fundamentally related to their true personal exposure distribution, than some repeated measurements of their exposure. 
Hence, we only consider~\eqref{gp_dr} as the analysis model and it \blue{will always be misspecified 
as $\mathbb{E}_{\wzi} f_0$ will not equal $\mathbb{E}_{Z_i} f_0$ for all $f_0$ unless $Z_i$ are exactly observed.} 

Since $f$ follows a Gaussian process prior, $\sum_j f(\bx_{ij})/m_{i}$ are jointly Gaussian distributed. Hence, from an implementation perspective, this becomes a standard Bayesian Gaussian process regression model that can be implemented using off-the-shelf software. We can also exploit the conjugacy, and obtain the closed-form posterior $f$. Specifically, stacking up the outcomes $\bY=(y_1,\ldots,y_n)^\top$ and $\bX=\{ \{\bx_{ij}\}_j : i \}$, for given $\sigma^2$ \blue{and kernel parameters}, we have:
\begin{align}\label{post_mean}
	\begin{split}
		\mathbb{E}[f(\bs)\ |\ \bX, \bY] &= \bm{l}^T(s)(\bM + \sigma^2\bI)^{-1} Y, \\
		\text{Cov}[f(\bs)\ |\ \bX, \bY] &= K(\bs, \bs) - \bm{l}^T(\bs) (\bM + \sigma^2\bI)^{-1} \bm{l}(\bs) 
	\end{split}
\end{align}
where $\bm{l}_i(\bs) = \sum_j K(\bx_{ij}, \bs)/m_{i}$ and $M_{ij} = \sum_{uv} K(\bx_{iu}, \bx_{jv})/(m_{i}m_{j})$. \blue{If a full Bayesian implementation is pursued,} adding 
priors for $\sigma^2$ \blue{and the kernel parameters, (\ref{post_mean}) gives the mean and variance of a normal distribution corresponding to the full conditional distribution for $f$} in a Gibbs sampler. 

We \blue{refer to our method as Gaussian process distribution regression (GPDR)} and note some properties of this model. \blue{First, the GP distribution regression (\ref{eq:gpan2}) linear in the sense that the relationship between the response $y_i$ and the distribution $Z_i$ is via a linear functional $\mathbb E_{Z_i} f$ of the distribution, even though we allow the regression coefficient $f$ to be an arbitrary smooth non-linear function by endowing it with a GP prior. Hence, it cannot model a non-linear relationship between $Z_i$ and $y_i$. In Section \ref{sec:nonlin} we propose a non-linear extension.} By construction, the model is invariant to any ordering of the observed samples for a subject. 
Our model is also invariant to any invertible transformation of observed $\bx_{ij}$. If only $w_{ij} = h(\bx_{ij})$ are observed, where $h$ is some unknown invertible transformation. Then we can implement model \eqref{gp_dr} with $\{w_{ij}\}$ by rewriting $y_i \sim \mathcal{N}\left(\sum_j \tilde f(w_{ij})/m, \sigma^2\right)$ where $\tilde f = f\circ h^{-1}$ and putting a GP prior on $\tilde f$. Therefore, our model is fairly robust to fixed systematical bias or scaling in measurements of covariates. 

\blue{One alternative to distribution model (\ref{gp_dr}) would be to have a non-linear regression model connecting each sample $\bx_{ij}$ to $y_i$, i.e., $y_i = f(\bx_{ij}) + \varepsilon_{ij}$ for $j=1,2,\ldots,m_i$, with $f$ modeled as a GP. This approach can thus be implemented simply as a Bayesian GP regression with $nm$ inputs $\{\bx_{ij}\}_{i,j}$. However, there are several reasons to prefer our distributional approach (\ref{gp_dr}) over this. First, as mentioned before, the premise of distribution regression is that the outcome does not relate to any of the measurements $\bx_{ij}$ individually but rather with the distribution $Z_i$ from which the $\bx_{ij}$ are drawn. Also, for subject $i$ this will lead to $m_i$ different generative models for $y_i$, one for each sample $\bx_{ij}$. So the model does not correspond to a valid probability distribution and Bayesian inference (especially uncertainty estimates) will not be valid. 
Finally, from a computational perspective, GP regression with $nm$ inputs will involve $O(n^3m^3)$ computation rendering this infeasible even for moderately small $m$ and $n$ which leads to a large $mn$. For example, with $m=100$ and $n=100$, one of the moderate sample sizes considered in our numerical experiment, this approach will require storing and inverting $10000 \times 10000$ matrices which is infeasible in typical personal computing environments. Our approach is more principled, linking $y_i$ to the distributions, having a single generative model for $y_i$. It is also computationally more tractable, only involving $n \times n$ matrices.}



\section{Theory}
Gaussian process regression has been provably successful in the setting of scalar- or vector-valued covariates \eqref{eq:gpreg}. \citet{van2011information} shows the posterior risk for the $L_2$ norm (estimation error for the regression function):
\begin{equation}\label{eq:risk}
\mathbb{E}_{f_0} \int \norm{f - f_0}^2_2 d \Pi(f | \bx_{1:n}, Y_{1:n})
\end{equation}
are bounded by an optimal rate, if the true regression function $f_0$ is \textit{H\"older smooth}. In this section, we show that the theoretical guarantees extend to our proposed Gaussian process regression for distribution-valued covariates. For the distributional setting, throughout our theory needs to accommodate the fact that our analysis model (\ref{gp_dr}) is misspecified with respect to the data generation process (\ref{eq:model}) on account of only observing the samples $\{\bx_{ij}\}$ instead of the distributions $Z_i$.

Two ideas are central to our theoretical results. First, as derived in Section \ref{sec:methods}, we can rewrite the  generative model from~\eqref{eq:model} for the outcome as $y_i = F_0(Z_i) + \eps_i$, where $F_0(Z_i)=\mathbb E_{Z_i}f_0$. Consequently, modeling $f$ as a GP on the support \blue{$\mathcal X$} of the covariates induces a GP $F(Z)=E_Z f$ on the space $\mathcal S$ of distributions in the support of $\mathcal Z$. This GP has a meta-kernel $\mathbb K$ derived in (\ref{eq:meta}) from the original kernel $K$. Thus, our distribution regression simply becomes a GP regression on the space of distributions, and we can invoke concentration results for GP regression on arbitrary linear spaces to control the posterior risk~\citep{van2008rates,van2011information}. 
However, these results on risk bounds do not account for the fact that $Z_i$ are unobserved, and the analysis model can only use the observed samples $\{\bx_{ij}\}$. Hence, the second key step is to decompose the actual risk for the analysis model into the risk when using the correctly specified model if $Z_i$ are observed (this risk is controlled via the aforementioned GP concentration results), and the excess risk arising from the model misspecification on account of using only the samples. We show that this excess risk can be controlled as the number of repeated measures $m_i$ grows. 

\subsection{Notations and Assumptions} \label{note_and_assump}
Suppose we have $n$ subjects \blue{with data generated from (\ref{eq:model}) where we only observed the outcomes $y_i$ and the $m_i$ covariate samples $\{\bx_{ij}\}$.} 
We would like to fit $f_0$ using the GP regression analysis model~\eqref{gp_dr}. For simplicity, we will assume $m_i = m$ for all $i$. But it is easy to see that all results will not change if we relax to $m_i \ge m$ for all $i$. We will denote $\mathbb{D}_n = \{\{\bx_{ij}\}_j, y_i\}_{i=1}^n$ to be the observed data and $\mathbb{Z}_n = \{Z_i, y_i\}_{i=1}^n$. 
We will show that 
the posterior risk $R_n$ based only on the observed data $\mathbb{D}_n$, contracts at an optimal rate. Here 
\begin{equation} \label{eq:postrisk}
R_n = \mathbb{E}_{f_0} \int ||f - f_0||^2 d \Pi_n(f|\mathbb{D}_n),   
\end{equation}
where $\Pi_n$ is the posterior distribution, $\mathbb{E}_{f_0}$ is relative to the distribution of $\mathbb{D}_n$, and $\norm{\cdot}$ can be the empirical norm $\norm{\cdot}_n$ or the $L_2$ norm $\norm{\cdot}_2$. The empirical norm is naturally defined as 
\begin{equation} \label{eq:enorm}
\norm{f}_n^2 = \frac{1}{n} \sum_{i=1}^n \left(\mathbb{E}_{Z_i} f\right)^2
\end{equation}
Note that the expectation in this norm is respective to the underlying distribution $Z_i$. 
If each $Z_i$ is a Dirac distribution with mass at $\bx_i$, (\ref{eq:enorm}) becomes the standard empirical norm $\frac 1n \sum_i f(\bx_i)^2$.

In the RKHS approach of solving functional linear regression, a method equivalent to the posterior mean estimation of a Gaussian process regression model, accessing the $L_2$ estimation error bound requires a restrictive assumption that the reproducing kernel commutes with the covariance kernel of the density process~\citep{yuan2010reproducing}, which is hard to be interpreted empirically. To derive theoretical guarantees of our approach, it is possible to isolate the assumptions on the reproducing kernel and the  distributional process. We present the following set of assumptions that have empirical interpretation.

To achieve the estimation error bound, we need the true function $f_0$ to be somewhat regular,  that is $f_0 \in \mathcal{F}$ for some functional space $\mathcal{F}$, 
with desirable smoothness properties. 
A function $f$ is $\alpha$-regular in $[0,1]^d$ if $f \in C^{\alpha}[0,1]^d \cap H^\alpha[0,1]^d$ for $\alpha > 0$ where $C^{\alpha}[0, 1]^d$ is the \textit{H\"{o}lder space} and $H^\alpha[0,1]^d$ is the \textit{Sobolev space} (see Section \ref{sec:def} of Supplementary Materials for definitions). $\alpha$-regular class will be the main functional class considered in our results.

For the GP prior, we consider the Mat\'ern family for the covariance kernel $K$ (defined in Section \ref{sec:def}), which is widely used in spatial statistics and non-parametric regression. The GP endowed with this kernel is called the Mat\'ern process. 
Sample paths of order $\alpha$ Mat\'ern process is $\alpha$-regular in the sense that it belongs to $C^\beta[0,1]^d \cap H^\alpha[0,1]^d$ for any $\beta < \alpha$ \citep{van2011information}.
We use $\mathcal{H}_K$ to denote the RKHS with kernel $K$ and when no confusion raises, $\mathcal{H}$ represents the RKHS of the GP prior, i.e., RKHS of the GP covariance kernel.  The corresponding RKHS norm is denoted by $\norm{\cdot}_\mathcal{H}$. Note that $K(s, t) \le \kappa$ for all $s, t$ for some fixed $\kappa \geq 1$ for all Mat\'ern kernel.


We also need assumptions on the distributional process $\mathcal{Z}$,  which generates the subject-specific distributions $Z_i$. 
Denote the mean measure as $\mu: \mu(A) = \mathbb{E}_{Z \sim \mathcal{Z}}[Z(A)]$. We assume $\mu$ has a density $\mu(x)$ that is bounded away from $0$. 
To interpret this assumption, consider the limiting case of $Z_i$ being degenerate at $\bx_i$, in which case the distribution regression simply reduces to the non-parametric regression~\eqref{eq:gpreg}. This assumption then translates to the underlying distribution generating the covariates $\bx_i$ having a density bounded away from $0$ \blue{on $\mathcal X=[0,1]^d$.} This is commonly assumed for studying GP regression~\citep{van2011information}. 
Under this assumption, without loss of generality, we can regard $\mu(x) = 1$ on $[0,1]^d$. This is because our model is invariant to any invertible transformation of $Z_i$. So we can always map to the 
uniform distribution on $[0,1]^d$ using the probability integral transform. 

Finally, in order to identify the regression function $f_0$, we will need properties that make $\mathcal{Z}$ be able to separate regular functions:
\begin{definition}
	A distributional process $\mathcal{Z}$ \textbf{weakly separates} a functional vector space $\mathcal{F}$ if and only if $\forall f_1, f_2 \in \mathcal{F}:\hspace{0.25em} \mathbb{P}_{Z\sim \mathcal{Z}}[\mathbb{E}_Z f_1 = \mathbb{E}_Z f_2] = 1 \Leftrightarrow f_1 = f_2$. And we call $\mathcal{Z}$ \textbf{strongly separates} $\mathcal{F}$ if and only if there exists a constant $C$ such that $\mathbb{E}_\mu f^2 \le C\mathbb{E}_{Z\sim\mathcal{Z}}[(\mathbb{E}_Z f)^2]$ for all $f \in \mathcal{F}$. Here $\mu$ is the expectation of $\mathcal{Z}$.
\end{definition}

In other words, weak separability asserts that if the distribution regression $\mathbb E_Z f_1$ is identical to $\mathbb E_Z f_2$ for two different regular functions and for almost all distributions $Z$ in the support of $\mathcal Z$, then $f_1=f_2$. This is reasonable, as without this it will not be possible to identify functions, given that they only enter the outcome model in the form of the expectation $\mathbb E_Z f$. Thus, weak separability is essentially the distributional analog of separability or full-rank assumptions of standard regression. \blue{Weak separability is satisfied by most distributional processes $\mathcal Z$ with a sufficiently rich support $\calS$ which can be set of all Dirac distributions on $\mathcal X$ or many families of continuous distributions with support on $\mathcal X$. See Supplementary Section \ref{sec:weak} for some examples.}

Strong separability is less intuitive and more technical. However, one can easily check that strong separability contains weak separability. This is because if $\mathbb{P}_{Z\sim \mathcal{Z}}[\mathbb{E}_Z f_1 = \mathbb{E}_Z f_2] = 1$, by strong separability, $\mathbb{E}_\mu (f_1 - f_2) ^2 \leq C \mathbb{E}_{Z\sim\mathcal{Z}}[(\mathbb{E}_Z(f_1 - f_2))^2] = 0$. Thus, $\mathbb{E}_\mu (f_1 - f_2) ^2 =0$ which implies $f_1 = f_2$.  
For the special case when $Z$ equals Dirac  distributions almost surely and our model reduces to the usual GP regression, it is easy to see that strong separability is satisfied with $C=1$. 
The following result presents a non-degenerate case where strong separability holds. 
\begin{lemma} \label{lem:dp}
	The Dirichlet process DP$(\blue{\upsilon}, \alpha)$, where $\alpha > 0$ is the concentration parameter,  strongly separates the space of bounded functions on $[0,1]^d$ and any measure $\blue{\upsilon}$ supported within it with $C = 1 + \alpha$ .
\end{lemma}

All proofs are provided in the Supplementary Materials. Lemma~\ref{lem:dp} shows that strong separability is satisfied by the Dirichlet process (DP), one of the most popular distributional processes (distribution on the space of distributions).   
Equipped with these assumptions, we now show that our proposed GP regression with distribution-valued covariates satisfy desirable error bounds. 

\subsection{Fixed Design}
We follow the notations and assumptions in Section~\ref{note_and_assump}. The outcomes $y_i$ and i.i.d samples $\{\bx_{ij}\}_j \sim Z_i$, for $i=1,\ldots,n$, are generated from model~\eqref{eq:model}. Our first result obtains bounds on the empirical norm for a fixed design. 
\begin{theorem}\label{th:fixed}
	Suppose the distributional process $\mathcal{Z}$ weakly separates $C^\beta[0,1]^d \cap H^\beta[0,1]^d$. Then for data generated from model~\eqref{eq:model} with any $\beta$-regular function $f_0$ and analyzed using model~\eqref{gp_dr} with an order $\alpha$ Mat\'ern kernel, there exists a constant $C$ independent of $n$ such that the posterior risk is controlled as
	\begin{equation} \label{res:emp}
	\mathbb{E}_{f_0} \int ||f - f_0||_n^2\ d \Pi_n(f|\mathbb{D}_n) \le C n^{-\frac{2\min(\beta, \alpha)}{2\alpha + d}}
	\end{equation}
	given $m = O(n^{C_{\alpha, \beta}})$, where $C_{\alpha, \beta}$ is a constant depending only on $\alpha$ and $\beta$. Minimax optimal rate is achieved when $\alpha = \beta$ and $m = O\left(n^{2 + \frac{d}{\beta} + \frac{4\beta}{2\beta + d}}\right)$.
\end{theorem}

For $\alpha=\beta$, our proposed distribution regression using the observed samples and a GP prior for the regression function  obtains optimal rates of function estimation. 
\blue{Note that this result is for fixed design, i.e., the norm is specified by the fixed distributional covariates $Z_i$. The result only requires weak separability and thus holds for most reasonable choices of $\calZ$ and $\calS$ with some examples provided in Section \ref{sec:weak}.}

\subsection{Random Design}\label{sec:random}
For the estimation error bound, one would consider $L_2$ norm corresponding to the mean measure $\mu$ of the distributional process $\mathcal{Z}$. For any $f$ that is $\mu$-measurable and with finite squared integral, we define $\norm{f}_2^2 = \int f^2 d\mu$. 
We now state the result on error bound for our Bayesian distribution regression for estimating the regression function $f_0$. 
\begin{theorem}\label{th:random}
	\blue{Assume that} the distributional process $\mathcal{Z}$ strongly separates $C^\beta[0,1]^d \cap H^\beta[0,1]^d$. Then for any $\beta$-regular function $f_0$ and for data generated from~\eqref{eq:model} and fitted with~\eqref{gp_dr} with order $\alpha$ Mat\'ern kernel for the GP,  there exists a constant $C$ independent of $n$ such that the posterior risk is controlled as
	\begin{equation} \label{res:est}
	\mathbb{E}_{f_0} \int ||f - f_0||_2^2\ d \Pi_n(f|\mathbb{D}_n) \le C n^{-\frac{2\min(\beta, \alpha)}{2\alpha + d}}
	\end{equation}
	given $m = O(n^{C_{\alpha, \beta}})$ and $\min(\alpha, \beta) > d / 2$, where $C_{\alpha, \beta}$ is a constant depend only on $\alpha$ and $\beta$. Optimal rate is achieved when $\alpha = \beta$ and $m = O\left(n^{2 + \frac{d}{\beta} + \frac{4\beta}{2\beta + d}}\right)$.
\end{theorem}
\blue{In Section \ref{sec:pfoutline} we provide the key ideas used in the proofs of Theorems \ref{th:fixed} and \ref{th:random} with technical proofs and details offered in Supplemental Section \ref{sec:proofs} (specifically Sections \ref{sec:def} through \ref{sec:banach}).} Once again, when $\alpha = \beta$, we get the minimax rate as $n^{-\beta / (2\beta + d)}$. 
For the special case when $Z$ equals Dirac  distributions almost surely, $\mathcal{Z}$ strongly separates all function space and our model degenerates to the typical Gaussian process regression~\eqref{eq:gpreg}. The optimal rate we can get is the same rate for estimating $\beta$-regular functions in a typical nonparametric regression case, which is $n^{-\beta / (2\beta + d)}$. This serves as a sanity check for our theory, as the distribution GP regression model is a generalization of the standard GP regression. 

Note that although it is required that $\alpha = \beta$ to achieve the optimal rate, we can actually require smaller $m$ with smaller $\alpha$ if the only goal is consistency. For example, from Section~\ref{sec:combine} of the Supplementary materials, we know that if $\beta > \alpha + d/2$, consistency only requires $m / n^2 \to \infty$, and $m / n^{2+d/\beta} \to \infty$ if $\alpha = \beta$.
\blue{Two factors contribute to this super-quadratic scaling of $m$ with respect to $n$ required for the theory. First, note that as the distribution $\widehat Z_i$ constructed from the observed samples $\{\bx_{ij}\}$ is not equal to the true distribution $Z_i$, the resulting posterior $\Pi_n(f \given \mathbb D_n)$ for model (\ref{gp_dr}) is different from the oracle posterior $\Pi_n(f \given \mathbb Z_n)$ one would have obtained from  (\ref{eq:gpan2}) using the true $Z_i$, and for which the existing theory of Gaussian process regression applies. The difference between the posteriors control the excess risk. From Section \ref{sec:dgp} we know that (\ref{eq:gpan2}) is GP regression with distributional inputs $Z_i$ and induced distributional kernel $\mathbb K$, defined in (\ref{eq:meta}). Then (\ref{gp_dr}) is GP regression with the same kernel $\mathbb K$ but with perturbed inputs $\wzi$. 
The kernel $\mathbb K$ is linear in the mean embedings of the distributions (introduced later in Section \ref{sec:kme}), 
In a GP regression with sample size $n$, and a linear kernel $\mathbb K$, a perturbation of each input 
by amount $\delta$ (in the spectral norm) leads to perturbation of amount $n \delta$ for the posterior mean and variances. Informally, the $O(n)$ scaling comes from bounding the inverse kernel matrix $\|(\mathbb K + \sigma^2 / n \bI)^{-1}  \| \leq n/\sigma^2$ which appears in the posterior mean and variance (see , e.g., (\ref{eq:var}) for a formal derivation for the variance perturbation). 
For $m$ i.i.d. samples, perturbations $\widehat Z_i$ from $Z_i$ are of the order $\delta=1/\sqrt{m}$ in terms of their mean embeddings (Lemma \ref{lem:kme}) thus leading to $O(n/\sqrt{m})$ perturbation in the posterior means and variances. This implies $m$ needs to scale atleast at rate larger than $O(n^2)$. The second part of the scaling for $m$ comes from approximating the true regression function $f_0$ by the RKHS $\mathcal H_K$ of the Mat\'ern GP kernel $K$ (of smoothness $\alpha$) which contains smoother ($\alpha + d/2$-smooth) functions \citep{kanagawa2018gaussian}. So if $\beta > \alpha + d/2$, i.e., when $f_0$ lies in the RKHS, we only require $m/n^2 \to \infty$ suffices for consistency. When $\beta \leq  \alpha + d/2$, the additional scaling for $m$ comes from standard results of approximating a less smooth $f_0$ by $\mathcal H_k$ \citep{van2011information}.}

\blue{The result requires strong separability so will hold when $\calZ$ is a Dirichlet process as proved in Lemma \ref{lem:dp}. 
However, strong separability may not hold if $\calZ$ is supported on collections of continuous distributions. However, our empirical results show that the method indeed performs well in terms of estimating the regression function $f_0$ even for continuous distributions.}

\subsection{\blue{Dependent samples}}\label{sec:dep}

\blue{Theorems \ref{th:fixed} and \ref{th:random} assume data generated from (\ref{eq:model}) which stipulates that the samples $\{\bx_{ij}\}$ for a subject are i.i.d. draws from the true-subject specific distribution $Z_i$. In many instances, $\{\bx_{ij}\}$ will possess some dependence structure. For example, if $\{\bx_{ij}\}$ are time-indexed measurements of a subjects exposure, it is likely that they will exhibit some form of serial correlation. We now extend the theoretical results to accommodate dependence among the samples. 

\begin{theorem}\label{th:dep}
Consider data generated from model (\ref{eq:model}) with the only change being that $\bx_{ij}$ are identical draws of $Z_i$ but are not independent, instead each $\{\bx_{ij}\}_j$ is an absolutely regular ($\beta$-) mixing sequence with mixing coefficients $\beta_i(k)$ for lag $k$. Then the results of Theorems \ref{th:fixed} and \ref{th:random} holds if $\sup_i \sum_k \beta_i(k) < \infty$. 
\end{theorem}

Proof of Theorem \ref{th:dep} is provided in Section \ref{sec:pfdep}. Theorem \ref{th:dep} shows that the posterior concentration results for GP distribution regression holds with the same rates even when the samples exhibit $\beta$-mixing class of dependence as long as the $\beta$-mixing coefficients satisfy  $\sup_i \sum_k \beta_i(k) < \infty$. The summability condition $\sum_k \beta_i(k) < \infty$ implies the coefficients $\beta_i(k)$ to decay at super-linear rate in terms of the lag $k$. The supremum over $i$  implies that this summability needs to be uniform across subjects. Many common temporal dependence structures satisfy this. For example, \cite{mokkadem1988mixing} showed that autoregressive (AR) and ARMA time series are $\beta$-mixing with geometrically decaying mixing coefficients $\beta_i(k)=\psi_i^k$ for some $\psi_i < 1$ and will be thus summable. 
Theorem \ref{th:dep} then ensures that the GP distribution regression still has optimal concentration rates for AR or ARMA dependence among samples as long as $\psi_i$s are bounded strictly below $1$.}

\subsection{\blue{Related theoretical work}}\label{sec:threview}

\blue{We briefly review existing theoretical support for scalar-on-distribution regression to put our theoretical contributions into context. We first note that while some scalar-on-distribution regression methods \cite{talska2021compositional,poczos2013distribution,oliva2014fast} can be formulated as scalar-on-function regression with pre-estimated density functions and that there is a vast theoretical literature for functional regression, there is a fundamental difference in the setup of functional and distribution regression which makes the theoretical studies different. In scalar-on-function regression, even if no covariate is truly functional, one actually observes the true function at a discrete set of points. In distributional regression, the true density function (presuming a density exists, which is not necessary) is never observed at any point. The discrete samples are just draws from the same distribution which is related to the response. This difference is important, as while there is theory for scalar-on-function regression that accounts for discrete realizations of the functions \cite{yao2005functional,crambes2009smoothing}, it does not apply to the distributional setting. 

To our knowledge, our theoretical results represent the first work on information rates of a Bayesian methodology for scalar-on-distribution regression. Previously, the theoretical properties of various non-Bayesian scalar-on-distribution regression methods has been studied in \citet{poczos2013distribution,oliva2014fast,szabo2016learning}. Each work presents a different methodology and somewhat different framework (e.g., assumptions) for studying the theoretical properties so direct comparison of the results is challenging, however, we provide discussion on some aspects of similarity and differences with our theory. 

\citet{poczos2013distribution} proposes a {\em kernel-kernel-estimator}, first estimating subject specific densities using KDE, then using the estimated densities as covariates in a non-linear regression using kernels on density spaces. They provide error bounds for prediction under various scenarios of scaling for the within-subject ($m$) and between-subject ($n$) sample sizes. For the setting nearest to ours, where, in our notation, they allow $m$ to be larger than $n$, the minimax rate is attained when $m = O(n^{\frac{\beta + \tau + 1}{2\beta + \tau}(d+2)})$ \citep[page 513, ][]{poczos2013distribution} where $\beta$ denotes a measure of smoothness of the function relating the density to the response (like ours), and $\tau$ is a {\em doubling dimension} controlling the complexity of the distributional process $\calZ$ . If $\tau \geq \beta$, the scaling can be much worse than ours as $d$, the dimension of the covariate space, can be quite large. E.g., if $d > 1$ this will lead to  super-cubic scaling for $m$. Ofcourse, they consider a fully non-linear setup and is thus accommodating a much larger class of data generation mechanisms than our linear paradigm. However, they rely on bounded errors which is quite restrictive. Assumption of a finite doubling dimensional for $\calZ$ also rules out discrete distributions. Finally, they allow the kernel bandwidths to depend on the sample size, which we do not. Optimizing over bandwidths have been shown to sharpen Gaussian process regression based error bounds \cite{yang} and can possibly further improve our scaling. 

\citet{oliva2014fast} estimate the densities based on observed samples using basis functions and then uses the estimated basis function coefficients as covariates for modeling the response. Their estimates of densities are not guaranteed to be densities. Also, their theory only considers a truncated estimator instead of the actual one and the rate of convergence scales as $O(m \log(m) / n \log(n))$. It is unclear why the rate worsens with increase in the number of samples $m$. This is different to our rates as well as rates in \cite{poczos2013distribution,szabo2016learning} which decrease with $m$ as more samples carry more information about the true distribution. 


\citet{szabo2016learning} presents the most comprehensive theoretical treatment for scalar-on-distribution regression, although they also restrict their analysis to bounded errors. They consider a general class of models, also based on two kernels akin to \cite{poczos2013distribution}, but using the first (inner) kernel to obtain distributional mean embeddings 
instead of KDE and using these mean embeddings as covariates, linking them to the responses via a second (outer) kernel regression. 
They study a multitude of models of varying complexity depending on the choice of the outer kernel and discuss the tradeoffs of the relative scaling of the two sample sizes, complexity of the true function class and the distributional process. 
We show in Section \ref{sec:kme} that our distribution regression corresponds to kernel mean embedding using the Mat\'ern kernel $K$ and then using a linear outer kernel on the mean embeddings. The scenario of \citet{szabo2016learning} that is closest to ours is when their outer kernel is Lipschitz of order $h=1$ \citep[Eq. 12,][]{szabo2016learning}. While they offer a sub-quadratic scaling of $m$ in terms of $n$ for optimal rates \citep[Remark 6,][]{szabo2016learning}, we believe this difference with our super-quadratic scaling primarily comes from differences in scaling of the regularization term. As discussed in Section \ref{sec:random}, the super-quadratic scaling for our method results from the effect of perturbation of inputs $Z_i$ to $\wzi$ in a GP regression. This effect can be bounded as the product of bound on the inverse regularized GP kernel matrix which is $\|(\mathbb K + \sigma^2/n I)^{-1} \| \leq  n/\sigma^2$ and the bound on the amount of perturbation for the mean embedding which is $O(1/\sqrt m)$ from Lemma \ref{lem:kme}. Writing $\lambda=\sigma^2/n$, this leads to the $O(\lambda^{-1}/\sqrt m)$ product bound necessitating $m$ to scale as larger than $\lambda^{-2}$ for the rate to vanish asymptotically. In ours and any typical Bayesian GP regression, the error variance $\sigma^2$ does not depend on $n$ and for a sample size of $n$ and the value of the regularization parameter $\lambda=\sigma^2/n$ naturally comes from the negative log-likelihood of the joint posterior, thus requiring $m$ to scale atleast larger than $O(n^2)$. The approach of \cite{szabo2016learning} 
optimizes a ridge-regularized loss function for the kernel regression. Although this loss has the same form as the negative log-Bayesian posterior, their regularization parameter $\lambda$ is not tied to any error variance model (as it is not a Bayesian method) and its scaling can be chosen to optimize the error rates. They have the same bounds for the regularized kernel matrix and the perturbation and thus 
they also require the same scaling for $m$ i.e., $\lambda^{-2}/m \to 0$ as seen from plgguing in $h=1$ in Eq (25) of \cite{szabo2016learning}. However, they allow $\lambda$ to scale at a rate slower than $O(1/n)$ \citep[Theorem 5,][]{szabo2016learning}. This translates to a sub-quadratic scaling for $m$ in terms of $n$. An equivalent choice for us would be to let $\sigma^2$ grow with $n$, which would allow $\lambda = \sigma^2/n$ to decay slowly. This can perhaps improve the scaling of $m$ with respect to $n$. However, letting the error variance grow with the sample size is somewhat unnatural in a Bayesian paradigm so we do not pursue this here. 


A key difference of our theory from the existing work is that, as we propose a Bayesian methodology, our theory gives error bounds for the entire posterior distribution of the risk for estimating the regression function, as opposed to error bounds of point predictions considered previously. To our knowledge this is the first theoretical treatment of posteriors from Bayesian scalar-on-distribution regression. We obtained optimal error rates and, as argued above, we conjecture that it might be challenging to further improve the scaling of $m$ that attains these optimal rates without additional assumptions or choices 
like letting the error variance or the kernel bandwidth to vary with $n$. Finally, none of the aforementioned theoretical work considered dependence among the observed samples, restricting to the i.i.d. case. We show in Theorem \ref{th:dep} that the error rates are attainable even under certain forms of within-subject dependence like autocorrelation.}


\section{\blue{Proof outline}}\label{sec:pfoutline}
We present the key arguments here for proving  \blue{Theorems \ref{th:fixed} and \ref{th:random}} while the technical proofs are in the supplement.

\subsection{Risk Decomposition}\label{sec:risk}
The \blue{critical step for the theory risk is decomposition of the risk (\ref{eq:risk}).} 
Let $\Pi_n(f|\mathbb{Z}_n)$ \blue{denote the posterior from the {\em oracle} GP regression (\ref{eq:gpan2}) if the true distributions $Z_i$ were observed.} Then the risk term (\ref{eq:postrisk}) can be decomposed into:
\begin{equation} \label{eq:decomp}
\begin{split}
R_n &= R_n^0 + R_n^1 \\
&= \mathbb{E}_{f_0} \int ||f - f_0||^2 d \Pi_n(f|\mathbb{Z}_n) + \\
& \qquad \left( \mathbb{E}_{f_0} \int ||f - f_0||^2 (d \Pi_n(f|\mathbb{D}_n) - d \Pi_n(f|\mathbb{Z}_n)) \right)
\end{split}
\end{equation}
where $\norm{\cdot}$ can be empirical norm $\norm{\cdot}_n$ and $L_2$ norm $\norm{\cdot}_2$.
\blue{Recall that the Gaussian process distribution regression (\ref{eq:gpan2}) is simply Gaussian process regression on the space of distributions with the induced kernel $\mathbb K$ defined in (\ref{eq:meta}). So,} we can bound the first term $R^0_n$ using \blue{general posterior concentration results for Gaussian processes using a} similar method as in \cite{van2011information}. 

\blue{The second term $R_n^1$ constitutes the excess risk on account of using a misspecified model (\ref{gp_dr}).} We bound $\mathbb{R}^1_n$ through a direct computation \blue{based on kernel mean embeddings of distributions which we discuss below. We emphasize that the risk decomposition (\ref{eq:decomp}) is considered for studying the theoretical properties of the method. In our method, we do not require knowledge of the true $Z_i$, neither do we attempt to model it, as we work directly with the samples $\wzi$.} 

\subsection{\blue{Kernel mean embeddings}}\label{sec:kme}

Let $\mathcal{H_K}$ be the RKHS of the GP kernel $K$ with inner product $\langle\cdot, \cdot\rangle$. \blue{For simplicity, we will often denote $\mathcal H_K$ simply as $\mathcal H$.} Let 
\begin{equation}\label{eq:kme}
\mu_Z (s) = \int K(\bs,\bt) Z(d\bt) 
\end{equation}be the kernel mean embedding of a distribution $Z$. 
\blue{The distributional $\mathbb K$ in (\ref{eq:meta}) can be simply expressed a second-stage linear kernel with the kernel mean embeddings as inputs, i.e., 
\begin{equation}\label{eq:linearkme}
\mathbb{K}(Z, Z') = \langle \mu_{Z}, \mu_{Z'}\rangle_\mathcal{H}.
\end{equation}
The difference between the posterior from the {\em oracle} Gaussian process regression (\ref{eq:gpan2}) had the $Z_i$ been known, and one from the direct approach (\ref{gp_dr}) actually used in practice,  simply results from the difference in GP regression with this linear kernel when replacing the true inputs $\mu_{Z_i}$ with the observed, perturbed inputs $\mu_{\wzi}$. 

Consequently, the key to control the excess risk $R_n^1$ between the two posteriors is to control the extent of the perturbation $\mu_{Z_i}$ to $\mu_{\wzi}$. Lemma \ref{lem:kme} quantifies this.

\begin{lemma}\label{lem:kme} $\mathbb E \|\mu_{Z_i} - \mu_{\wzi}\|^2_{\mathcal H} \leq 4\kappa/m$.
\end{lemma}
}

\blue{Recall that $\kappa$ is the upper bound on the kernel $K$. The result shows that the distance between the mean embeddings is $O(1/m)$, thus decreasing with the number of samples $m$. This is expected as with more samples collected, $\wzi$ becomes a better approximation of $Z_i$. As we outline below, the impact of perturbation of the distributional inputs from $Z_i$ to $\wzi$ propagates to the excess risk term $R_n^1$ solely through the expected distance between the mean embeddings given in Lemma \ref{lem:kme}.} 

\subsection{\blue{Controlling $R_n^1$}}\label{sec:perturb}
\blue{We show in Section \ref{sec:morder} that  we can have the following decomposition for $R_n^1$: 
\begin{align}\label{eq:decomp}
	R_n^1 = \mathbb{E}_{f_0} \int_{[0,1]^d} \left[  V(s) +   E_2(s)  - 2f_0(s) E_1(s) \right] ds
\end{align}}
where $V(s) = \text{Var}(f(s)|\mathbb{D}_n) - \text{Var}(f(s)|\mathbb{Z}_n)$, $E_1(s) = \mathbb{E}_{\bm{\varepsilon}} \left[\mathbb{E}(f(s) | \mathbb{D}_n) - \mathbb{E}(f(s) | \mathbb{Z}_n)\right]$ and $E_2(s) = \mathbb{E}_{\bm{\varepsilon}} \left[\mathbb{E}^2(f(s) | \mathbb{D}_n) - \mathbb{E}^2(f(s) | \mathbb{Z}_n)\right]$. \blue{The terms $V(s)$, $E_1(s)$ and $E_2(s)$ are all based on the difference between the expected posterior means and variances corresponding to the two GP distribution regressions (\ref{eq:gpan2}) and (\ref{gp_dr}). As the induced distributional GP kernel  is linear in the mean embeddings (\ref{eq:linearkme}), bounds for all three terms are determined by the expected distance between the mean embeddings established in Lemma \ref{lem:kme}. These bounds are established in Sections \ref{sec:v} through \ref{sec:e2}.}

Combining all \blue{the 3 bounds} we have:
\begin{equation} \label{eq:r1rate_in}
R^1_n = O\left( \frac{n}{\sqrt{m}}  \right)
\end{equation}
if $f_0 \in \mathcal{H}_K$ ($\beta \ge \alpha + d/2$). And if $f_0 \not\in \mathcal{H}_K$ we would get a more complicated and worse rate as:
\begin{equation} \label{eq:r1rate_out}
R^1_n = O\left(\frac{ n^{1 + \frac{(2\alpha - 2\beta + d)\gamma}{2\beta}}}{\sqrt{m}} + \frac{n^{2-\gamma + \frac{(2\alpha - 2\beta + d)\gamma}{2\beta}}}{\sqrt{m}} + \frac{n^{3-2\gamma}}{\sqrt{m}} \right)
\end{equation}
given $n / \sqrt{m} \to 0$. Therefore, when $\alpha = \beta$, setting $\gamma = 1$ we get that $R^1_n = O(n^{-2\beta / (2\beta + d)})$ if $m = O\left(n^{2 + \frac{d}{\beta} + \frac{4\beta}{2\beta + d}}\right)$.

\subsection{\blue{Controlling $R^0_n$}}\label{sec:r0}

\blue{The term $R^0_n$ in (\ref{eq:decomp}) is only based on model (\ref{eq:gpan2}) using the true distributions $Z_i$. Using the induced kernel (\ref{eq:meta}) we can rewrite (\ref{eq:gpan2}) directly as a GP regression on distributional inputs
\begin{align}\label{eq:gpan2dist}
\begin{aligned}
y_i &= F(Z_i) + \varepsilon_i, \, F \sim GP(0,\mathbb K(\cdot,\cdot)), \, \varepsilon_i \iid N(0,\sigma^2).
\end{aligned}
\end{align}

Thus $F$ is a Gaussian random element taking values in a linear space $\mathcal B$ such that there exists an bijection $\pi$ from $\mathcal{B}$ to $\mathcal{C}_0$: the bounded continuous function on $\calX=[0, 1]^d$ and $\forall F \in \mathcal{B}, Z \in \mathcal{S}: \ F(Z) = \mathbb{E}_Z \pi(F)$. The true generative model (\ref{eq:dgp}) then can be written as $\mathbb E_{\varepsilon_i} y_i = F_0(Z_i)$ where $\pi(F_0)=f_0$. 

The bijection $\pi$ can be used to establish that $\calB$ is a separable Banach space allowing us to invoke posterior concentration results for GP regression when the GP takes value in such spaces \citep{van2008rates,van2008reproducing}. This offers a bound for the term $R_n^0$. See Section \ref{sec:banach} for details. }

\section{\blue{Implementation and} extensions}
\subsection{\blue{Bayesian implementation}}\label{sec:bayes}
\blue{Implementations of our proposed GPDR vary in computational complexity and degree of approximation or simplification. A full Bayesian implementation of the model (\ref{gp_dr}) assigning priors on both $\sigma^2$ and kernel parameters $\phi$ is free of any approximation (beyond using Markov Chain Monte Carlo to sample from the posterior). However, } posterior inference from a Gaussian process regression like~\eqref{gp_dr} is typically slow when we have a large sample size, since the computations require $O(n^3)$ time complexity with $n$ be the number of subjects. \blue{This is further complicated when assuming unknown kernel parameters as this $O(n^3)$ computation -- related to inversion of the $n \times n$ GP kernel matrix needs to be repeated within every MCMC iteration. This will be prohibitive even for moderate values of $n$. A compromise would be to fix the kernel parameters, a common practice in kernel regression, which then only requires a one time inversion of the GP kernel matrix in the entire MCMC run. The full Bayesian implementation is also amenable to extensions like fully non-linear GPDR (as outlined in Section \ref{sec:nonlin}) and to non-Gaussian responses. Also, this Bayesian implemention of GPDR can be done using off-the-shelf software like STAN requiring only specification of model (\ref{gp_dr}) and all the priors. We implement this Bayesian version in STAN and demonstrate the accuracy of GPDR in one of the numerical experiments.}

\subsection{Low Rank Approximation} \label{sec:low_rank}

\blue{The computational challenges of Gaussian process regression exacerbate} in our distribution regression setting, since every subject can involve multiple samples, and the resulting posterior mean and covariance operators given in~\eqref{post_mean} require summing over all sampled covariate values of each subject. If using a \blue{Bayesian implementation as outlined in Section \ref{sec:bayes},} the overall time complexity would be $O(n^3 + n^2m^2)$, where $m$ is the number of samples for each subject. \blue{This becomes prohibitive when either $n$ or $m$ is moderately large.}

\blue{We here propose a fast alternative using MCMC-free conjugate Bayesian updates and a low-rank Gaussian process approximation. If the kernel parameters and $\sigma^2$ are held fixed, leveraging conjugacy of the Gaussian likelihood for $y$ with the Gaussian process prior for $f$, the posterior of $f$ is available in closed form as specified in (\ref{post_mean}). Conjugate Bayesian methods are becoming increasingly popular as fast alternatives for GP regression \citep{finley2009improving} as they perform at par with fully Bayesian methods while being orders of magnitude faster \citep{heaton2019case}.} 

\blue{The conjugate updates (\ref{post_mean}) for GPDR still requires the one-time $O(n^3)$ computation for inverting the full rank the GP kernel matrix.} 
Many excellent methods to scale-up GP computations are now available \citep[see][ for a review]{heaton2019case}. 
In particular, low- or finite-rank GP approximations \citep{banerjee2008gaussian,cressie2008fixed,finley2009improving, datta2016hierarchical} provide excellent scalable inference. We  
propose the following low-rank approximation technique to accelerate our algorithm. 

From~\eqref{post_mean}, using the Representer Theorem, we consider representing $f$ as:
\begin{equation} \label{f_present}
f(s) = \sum_{i=1}^n w_i l_i(s), \hspace{1em} \text{with } \bm{w} =(w_1,\ldots,w_n)^\top \sim \mathcal{N}\left(\bm{0},  \bm{M}^{-1}\right)
\end{equation}
with basis $l_i(\bs) = \sum_j K(\bx_{ij}, \bs)/m_{i}$ and the covariance matrix given by $\bm{M}_{ij} = \sum_{uv} K(\bx_{iu}, \bx_{jv})/(m_{i}m_{j})$. The prediction, $\hat{y}_i = \sum_j f(\bx_{ij})/m_{i} = \sum_j \bm{M}_{ij} w_j$, and we have the matrix representation as $\hat{\bm{y}} = \bm{M} \bm{w}$. Therefore, the posterior mode for $\hat{\bm{w}}$ could be find through minimizing
\begin{equation} \label{eq:minimize}
\norm{\bm{y} - \bm{M} \bm{w}}^2 / \sigma^2 + \bm{w}^\top \bm{M} \bm{w}
\end{equation}
Let $\bm{M} = \bm{U}\bm{D}\bm{U}^\top$ be the eigendecomposition of $\bm{M}$. For a fixed $k \ll n$, we approximate $\bm{M}$ as $\bm{U}_k\bm{D}_k\bm{U}_k^\top$, where $\bm{U}_k$ consists of the first $k$ eigenvectors and $\bm{D}_k$ is the diagonal matrix of first $k$ eigenvalues. Restricting $\bm{w}$ be the column space of $\bm{U}_k$ with $\bm{w} = \bm{U}_k \bm{w}_k$, ~\eqref{eq:minimize} becomes
\begin{equation} \label{eq:low_rank}
\min_{\bm{w}_k} \norm{\bm{y} - \bm{U}_k\bm{D}_k \bm{w}_k}^2 / \sigma^2 + \bm{w}_k^\top \bm{D}_k \bm{w}_k
\end{equation}
Such a low-rank approximation~\eqref{eq:low_rank} is very similar to the thin plate spline method~\citep{wood2003thin}. Therefore, the whole process can be efficiently implemented using a spline regression package, such as \textit{mgcv} in R language. Through such an approximation, we could reduce the computation complexity of the regression part from $O(n^3)$ to $O(nk^2)$ using a suitable Lanczos algorithm. The $O(m^2)$ part for computing each entry of $\bm{M}$ remains the same. 
However, since these entries are simply averages over the repeated measures, the computation could be fully parallelized. Also, when $m$ is really large, down-sampling or binning the observation $\bx_{ij}$ into tractable size will further speed up the algorithm while yielding a reasonable approximation.

The low-rank approach has some similarities with the Bayesian distribution regression in \citet{law2018bayesian}. However, they use the fixed landmark points method, i.e., they are representing the regression function $f$ directly as $f(s) = \sum_{i=1}^R w_i K(\bs_i, \bs)$ with manually  selected and fixed landmark points $\{\bs_i\}$. By using a limited number of landmark points, they can control the size of the $\bm{M}$ covariance matrix. 
However, the Representer theorem suggests equation~\eqref{f_present} is the correct basis expansion. So the approach in \cite{law2018bayesian}, using fixed landmarks to express the function, is hard to justify theoretically. We thus expect our method to offer a better approximation and observe this in empirical studies.

\subsection{\blue{Non-linear extension}}\label{sec:nonlin}
\blue{As discussed in Section \ref{sec:an}, the GP distribution regression (\ref{gp_dr}) is a linear model on the space of distributions. More formally, when (\ref{eq:gpan2}) is represented as a GP regression (\ref{eq:gpan2dist}) on the distributional outputs, we have $\mathbb E_{\varepsilon_i} y_i = F(Z_i)$ where $F \sim GP(0,\mathbb K)$. The induced distributional kernel $\mathbb K$ given in (\ref{eq:meta}) is linear in the mean embeddings as seen in (\ref{eq:linearkme}). This linear kernel is not an universal kernel (see \cite{micchelli2006universal} for a definition), i.e, the resulting RKHS $\mathcal H_{\mathbb K}$ does not have the universal approximating property of being able to approximate any continuous function on the space of distributions. For example, if the data is generated as $\mathbb E_{\varepsilon_i} y_i = F_0(Z_i)$ where $F_0(Z) = \mbox{median}(Z)$ or $F_0(Z) = (\mathbb E_{Z} f^*)^2$ for some function $f^*$, then there is no $f_0$ for which $F_0(Z)=\mathbb E_{Z} f_0$ for all $Z$ in any reasonably rich collection $\calS$ of distributions. See Supplementary Section \ref{sec:nonuni} for concrete examples and more detailed explanation of why $\mathbb K$ is not universal.

We present a simple non-linear extension of the GP distribution regression model. We consider data generated as
$\mathbb E_{\varepsilon_i} y_i = F_0(Z_i)$ 
where $F_0$ is some, possibly non-linear, functional linking $Z_i$ to $y_i$. A analysis model for such data will be a regression $\mathbb E_{\varepsilon_i} y_i = F(Z_i)$ where $F$ needs to be modeled to be from a sufficiently rich class of functionals in order to capture possible non-linearity.} 
We first motivate our direct approach using  permutation or order invariance. \blue{As we only observe the samples $\wzi$,} 
we consider regression using the entire set of samples, i.e., we consider $y_i = F(\wzi) + \varepsilon_i$. \blue{As $\wzi$ can be completely characterized by the samples $\{\bx_{ij}\}_j$ we can think of} $F$ is a mapping from set $\{\bx_{ij}\}_j$ to a real number. 

\blue{The premise of of distribution regression is that order of collection of the samples $\bx_{ij}$ is unimportant,} i.e., the index $j$ in the samples $\bx_{ij}$ does not represent any meaningful information. \blue{If that is not the case, then $\{\bx_{ij}\}$ should be considered as functional data and functional regression approaches will be more suitable. For distribution regression we would want the function $F$ to be invariant to the order of the samples $\bx_{ij}$.} 
\citet{zaheer2017deep} shows any such permutation-invariant function $F$ can be represented as $F(\{\bx_{ij}\}_j) = \phi\left(\sum_j f_i(\bx_{ij})\right)$ for some functions $\phi$ and $f_i$. The linear GP distribution regression (\ref{gp_dr}) then comes naturally with $\phi_i(x) = x$ (linear link) and $f_i = f/m_{i}$ with $f$ being the common non-linear regression function shared across all subjects. The normalizing factor, $1/m_{i}$, ensures that different numbers of samples contribute equally.

\blue{We can easily relax the assumption of linearity of the link to have a non-linear link function. Formally, we let $\phi_i=\phi$, a non-linear link function common to all the subjects (as subject-specific links cannot be identified with only one observed response per subject). In principle $\phi$ can also be modeled as a GP, in addition to modeling $f$ as a GP. Alternatively, $\phi$ can be modeled using a basis expansion, which we do here. This leads to the following single-index type non-linear GP distribution regression model. 

\begin{equation}\label{eq:nlgpdr}
\begin{split}
y_i  &= \phi\left(\frac 1{m_i} \sum_{j=1}^{m_i} f(\bx_{ij})\right) + \varepsilon_i, f \sim GP(0,K(\cdot,\cdot)), \, \varepsilon_i \iid N(0,\sigma^2),\\
\phi(x) &= \sum_{r=1}^R \gamma_r \phi_r(x),\, \bgamma = (\gamma_1, \ldots, \gamma_R)^\top \sim N(\bmu_\gamma, \bSigma_\gamma).
\end{split}
\end{equation}

Here $\phi_1,\ldots,\phi_R$ denote a set of known basis functions. The unknown basis coefficients  are stacked in $\bgamma$ which is assigned a conjugate multivariate normal prior. Like the full Bayesian version of the linear model described in Section (\ref{sec:bayes}), the non-linear GPDR can also be implemented by simply specifying the model (\ref{eq:nlgpdr}) in an off-the-shelf software to run Bayesian algorithms. We offer a proof-of-concept implementation of (\ref{eq:nlgpdr}) in STAN and provide a comparison of the linear and non-linear GPDR in one set of numerical experiments.}

\subsection{Incorporating other subject-specific information}\label{sec:subject}

The Bayesian hierarchical model framework~\eqref{eq:dgp} and~\eqref{gp_dr} we proposed for distribution regression is very flexible and can be easily extended for more structured data types with additional subject-specific information. For example, if there are other measured vector-valued covariates $W_i$ for each subject, it can be easily included along with the distribution-valued covariates, by including a linear regression term. Thus, we extend~\eqref{eq:dgp} to $\mathbb E_{\varepsilon_i} y_i = \beta_0^\top W_i + E_{Z_i} f_0$ and the analysis model~\eqref{gp_dr} similarly \blue{by adding a linear regression term with coefficient $\beta$.} Implementation will remain efficient, as the Gibbs sampler will have conjugate updates for the regression coefficient $\beta$, and conditional on $\beta$, the updates for $f$ will be similar to~\eqref{post_mean} but replacing $y_i$ with $y_i - \beta^\top W_i$. 

Similarly, if the data are clustered in nature, one can easily modify the outcome model to have $\mathbb E(y_{ic}) = \mathbb E_{Z_i}(f_{c})$, where $c$ denotes the $c^{th}$ cluster and $f_{c}$ denotes cluster-specific regression functions, which can be modeled as exchangeable draws from a Gaussian process, akin to standard random effect models for clustered data. If the data are functional in nature, but both the functional and distributional aspects are important, then one can extend the outcome model to have both a functional and a distributional regression component. Inference from all these models, like our base model, can be obtained using off-the-shelf Bayesian software. We will pursue these extensions in the future. 

\section{Simulations}\label{sec:sim}
One difference between our proposed \blue{Gaussian process distribution regression (GPDR)} to \blue{many} existing \blue{methods for distribution regression,} like \citet{poczos2013distribution,oliva2014fast}, is that we do not require estimating the underlying densities. Because the sample mean is the best estimator of the expectation in the non-parametric sense, we would expect our model to do better than functional linear models with estimated densities. Also, by using full samples, our model should have better asymptotic performance compared to the Bayesian models presented in \citet{law2018bayesian} with only a fixed amount of landmark points. In this section, we conduct a simulation study to show that our method converges and has a better convergence rate compared to the alternatives using either estimated densities or landmark points.

We simulate our data as
\begin{equation}\label{sim:ys}
Z_i \sim DP\left(\text{Unif}[0, 1], 25\right),\, \bx_{ij} \iid Z_i; \varepsilon_i \iid \mathcal{N}(0, 0.01) 
\; 
y_i = \mathbb{E}_{Z_i}(f_0) + \varepsilon_i,
\end{equation}
where $f_0(x) = 10x \cdot \exp(-5x)$ is infinitely smooth within $[0,1]$.
We draw $n$ by $m$ samples $\bx_{ij}$ and corresponding $y_i$ from (\ref{sim:ys}). Here $n$ ranges within $\{50, 100, 200, 300, 400\}$ and $m$ ranges within $\{50, 100, 250, 500, 1000, 2000\}$. We do the exact posterior inference using~\eqref{post_mean}, not the low-rank approximation,  and compare the empirical risk $\int \norm{f - \hat{f}}_2^2 d \Pi(\hat{f}\ |\ \{X, Y\})$ estimated with 100 samples from the posterior process $\Pi(\hat{f}\ |\ \{X, Y\})$ for each combination of $n, m$. For the \textit{Bayesian density regression (BDR)} model introduced in \citet{law2018bayesian}, we use 10 and 50 evenly spaced landmark points in $[0, 1]$ and set all other hyperparameters as default. We also compare our method with a direct density estimation alternative, that is, to replace empirical expectation $\sum_j f(\bx_{ij}) / m$ with $\int f(\bx) \widehat{dZ_i}(\bx) d(\bx)$ in~\eqref{gp_dr}, where $\widehat{dZ}_i$ is a density of $Z_i$ estimated from kernel density estimation (KDE). The KDE analysis model is thus given by 
\begin{align} \label{gp_dr_kde}
	y_i \sim \mathcal{N}\left(\int f(\bx) \widehat{dZ_i}(\bx) d(\bx), \sigma^2\right),\,
	f \sim GP(0, K).
\end{align}
It is easy to see that such an alternative is a functional linear regression in Reproducing Kernel Hilbert Space with the same kernel $K$ as in (\ref{gp_dr}).

\begin{figure}[!h]
	\centering
	\includegraphics[width=0.95\linewidth]{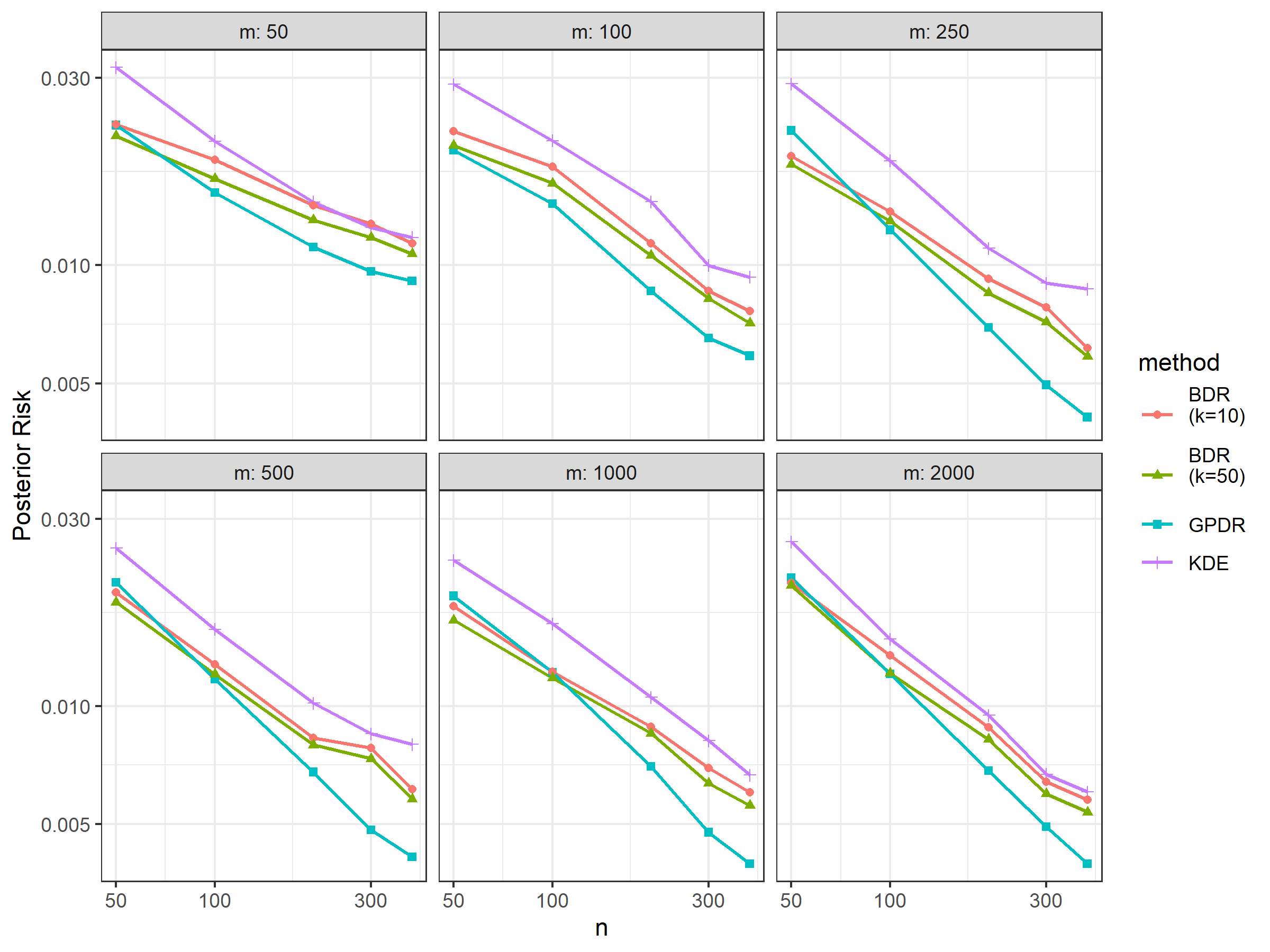}
	\caption{\label{fig:sim}The empirical posterior risk for every combination of $n$ and $m$ under 100 runs, where $n$ stands for the number of subjects and $m$ is the number of samples for each subject. GPDR is for our Gaussian process distribution regression model; KDE stands for its direct alternative~\eqref{gp_dr_kde} using kernel density estimated expectations; BDR stands for the Bayes distribution regression method suggested in \citet{law2018bayesian}, where $k$ is the number of landmark points.}
\end{figure}
We use a typical Gaussian kernel $K(s, t) = \exp\left(-(s-t)^2 / 2l\right)$, where $l$ is set to $0.25$. We run every setting 100 independent times and report the final mean empirical risk with confidence intervals for every setting. We show the results in Figure~\ref{fig:sim}. It is clear that the estimation error from our algorithm converges in polynomial rate when we increase $n$ with sufficient large $m$. This is at par with our theory. We see that our method is significantly better than its kernel density estimation alternative across all the settings. We also have a better convergence rate and smaller estimation error compared with the \textit{BDR} method in \citet{law2018bayesian} for large $n$, irrespective of the choice of landmark point numbers. This is not surprising since \textit{BDR} uses only a fixed number of basis functions and the gap is similar to that of fixed basis spline regression compared to Gaussian process regression. {\color{blue}We remove BDR from the further performance comparisons provided below, due to its slow running time.}



\subsection{\blue{Continuous distribution and dependent samples}}\label{sec:simcont}
{\color{blue} 
In the previous section, the distributions $Z_i$ were generated from a Dirichlet process to ensure that the strong separability condition is met (Lemma \ref{lem:dp}) and the data are generated under conditions for which Theorem \ref{th:random} on estimation consistency holds. Realizations of DP are almost surely discrete distributions, although the large concentration parameter $\alpha$ ensured that the distribution is nearly continuous (in the sense the CDF of $Z_i$ can be closely approximated in the supremum norm by continuous functions).

We also assess the generalizability of the approaches to truly continuous distributions to assess robustness to violation of the strong separability assumption. We generate data as follows: With $\logit (x) = \log (x/(1-x))$, we hierarchically simulate 
\begin{equation*}
    c_i \iid \mathcal{N}(0, 4), \,\, \logit (Z_i) = N(c_i, 0.09), \,\, (\bx_{ij}) \iid Z_i. 
\end{equation*}
Thus, $Z_i$ is a continuous (logit-normal) distribution. Here, the logit transformation is used to ensure the support of $Z_i$ is compact $([0,1])$. 

We also consider a case where the samples $\bx_{ij}$ are not independent and there is a serial within-subject correlation. For each subject $i$, we induce correlation between $\{\bx_{ij}\}_j$ to analyze the impact of dependent covariates. Let $\Phi = (\rho^{\abs{i-j}})$ denote the $m \times m$ correlation matrix under the AR(1) dependence with coefficient $\rho$. Then for each subject $i$, we simulate dependent covariates $\bx_i$ as
\begin{equation*}
    c_i \iid \mathcal{N}(0, 4),\,\, \logit (Z_i) = N(c_i, 0.09),\,\, \logit (\bx_i) \,|\, c_i \ind N (c_i \boldsymbol{1}, 0.09 \Phi),
\end{equation*}
where $\logit (\bx_i) = (\logit (\bx_{i1}), \cdots,\logit (\bx_{im}))^\textsc{T}$ and $\boldsymbol{1}$ is a vector of length $m$ with all entries 1. Generating $\bx_i$ this way ensures that each $\logit(\bx_{ij}) \sim N(c_i,0.09)$, i.e, $\logit(\bx_{ij}) \sim \logit(Z_i)$, i.e.,  $\bx_{ij} \sim Z_i$ but that $\bx_{ij}$ are not independent across $j$ and they exhibit AR(1) correlation in the logit scale. 

\begin{figure}[!h]
	\centering
	\includegraphics[width=.9\linewidth]{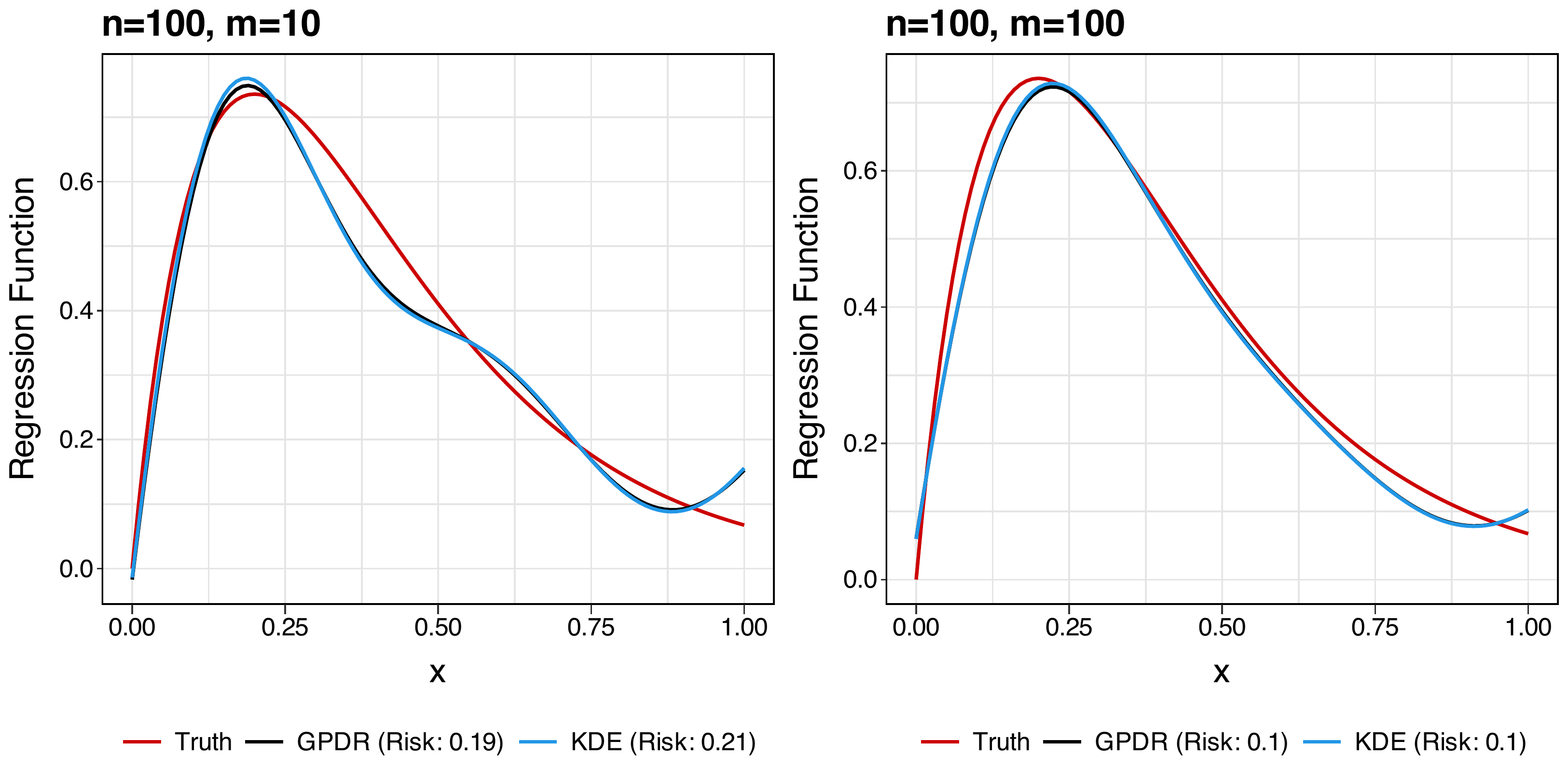}
	\caption{\label{fig:reg func CP dep}Given a simulated data with dependent continuous covariates for $\rho=0.5$, this compares regression function estimates from KDE and GPDR with the truth.}
\end{figure}

We consider three choices of $\rho$: $0$, $0.5$ and $0.8$ with $\rho=0$ corresponding to the setting of i.i.d. samples. We first show results for one simulated dataset for one of dependent settings with $\rho=0.5$. Figure~\ref{fig:reg func CP dep} compares regression function estimates using GPDR and KDE. We see that both methods perform similarly well in this case of continuous distributions and dependent samples, estimating the regression function accurately. The density estimate from KDE for samples generated from a continuous distribution closely represents the true distribution $Z_i$ well. So the performance of KDE for this scenario is much better than when the distributions are discrete. However, even for this scenario that is favorable to KDE and misspecified for the estimation theory of GPDR, it performs at par with KDE. Also, the results show even for dependent samples, GPDR can estimate the regression function accurately.
The average performance across $100$ replicate datasets for all different data generation scenarios also reveals GPDR to be consistently better or at par with KDE. This is presented in Figure \ref{fig:risk ratio} and discussed in detail later. 

\begin{figure}[!h]
	\centering
	\includegraphics[width=\linewidth]{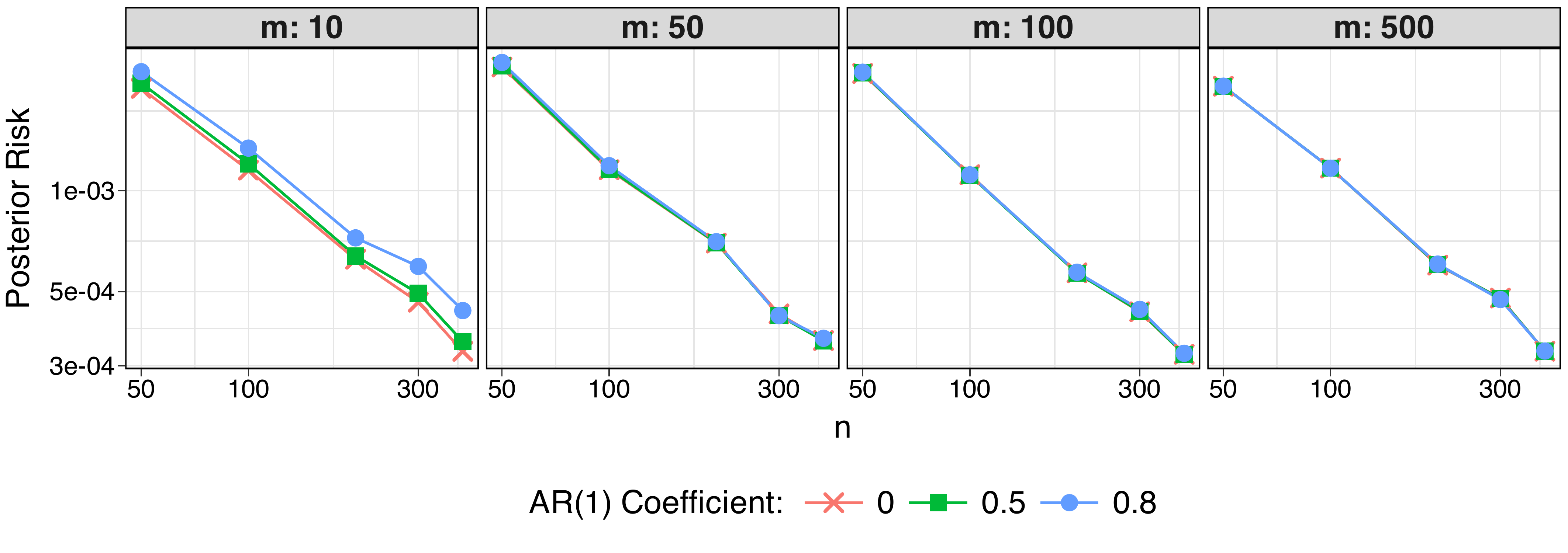}
	\caption{\label{fig:risk gpdr varying dependence} For different $n$ and $m$, this presents the empirical posterior risk of GPDR based on 100 simulation runs where $\rho$ is varied as 0 (independent), 0.5 (moderate), and 0.8 (high).}
\end{figure}
In Figure~\ref{fig:risk gpdr varying dependence}, we examine the impact of within-subject dependence on GDPR risk by varying the degree of dependence, represented by $\rho$, at three levels: 0 (independent), 0.5 (moderate), and 0.8 (strong). The figure reveals that for small values of $m$, the risk is higher when the dependence is strong.  However, as $m$ increases, the performance improves, and the risks become similar across all $\rho$ values tested. These findings collectively suggest that for moderate or large $m$, GPDR performs comparably to situations where covariates are uncorrelated while for small $m$ i.i.d. samples are desirable. 
The results underscore the reliability and robustness of GPDR even under within-subject covariate dependence, corroborating the theory for the dependent case in Theorem \ref{th:dep}.

\subsection{\blue{Improvement over kernel density estimation based methods}}\label{sec:simkde}
The simulation experiments summarized in Figure \ref{fig:sim} already showed the tangible superiority of GPDR over KDE. We now present a case of more extreme improvement afforded by GPDR over density estimation based methods. The data generation process given in Equation (\ref{sim:ys}) generated the distributional covariates as draws of a Dirichlet process with concentration parameter 
$\alpha=25$. This parameter $\alpha$ governs the degree of discreteness in the subject-specific distribution $Z_i$. As $\alpha \uparrow \infty$, they spread out across a wider range, exhibiting more continuous behavior. For $\alpha \downarrow 0$, the observed values $\bx_{ij}$ tend to cluster around a few distinct values. Figure~\ref{fig:DP 0.1 vs 25} gives an illustration by comparing empirical CDFs of $Z_i$ for $\alpha \in \{0.1, 25\}$, based on 100 random samples
\begin{figure}[!h]
	\centering
	\includegraphics[width=0.8\linewidth]{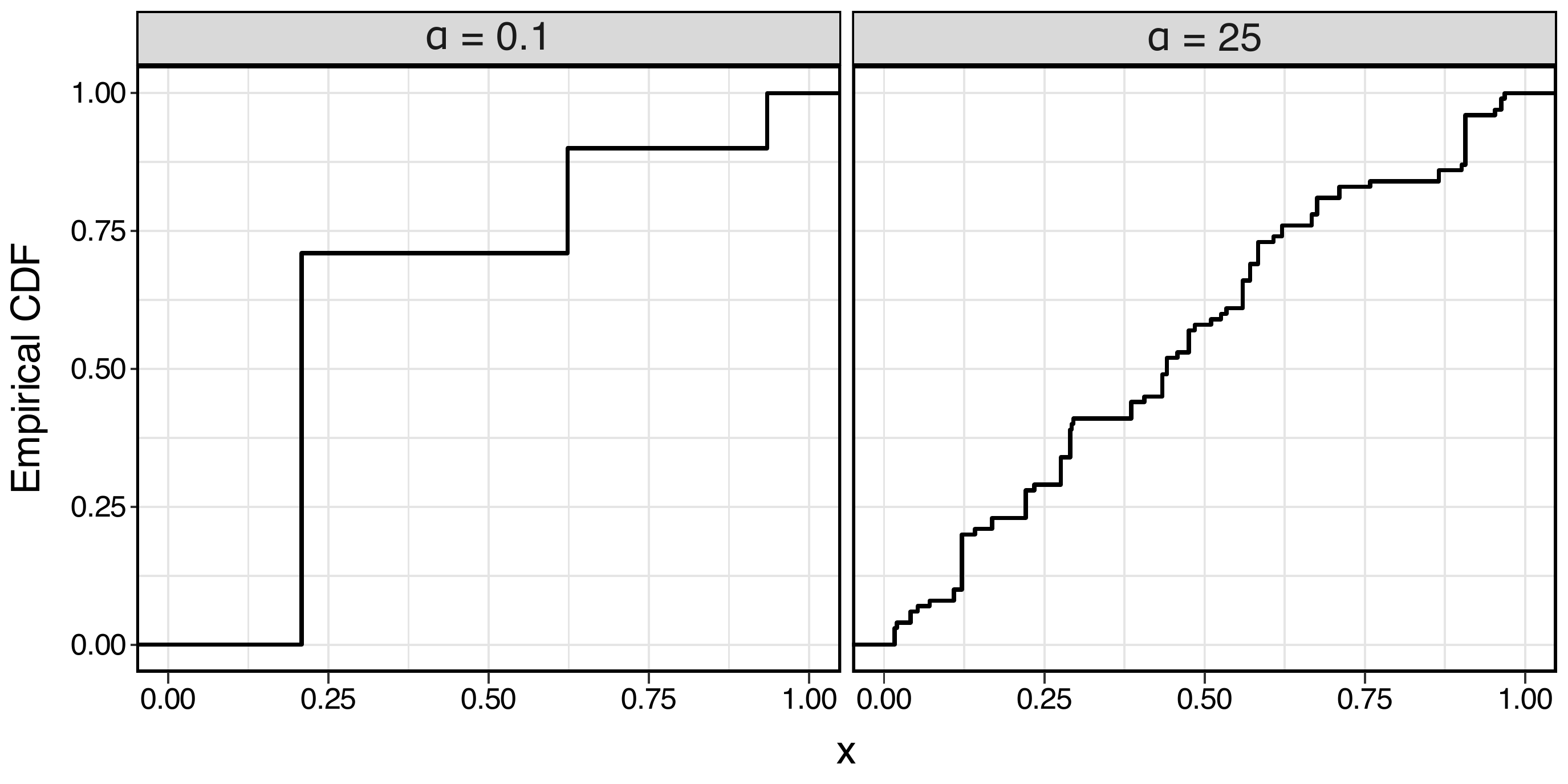}
	\caption{\label{fig:DP 0.1 vs 25}The empirical CDF of 100 random samples drawn from the Dirichlet Process with concentration parameters $(\alpha)$ 0.1 and 25.}
\end{figure}
  We see that for $\alpha=25$, which was used in (\ref{sim:ys}), the distributions are supported on many points with small point masses and the corresponding CDF is much closer to continuous CDFs (in terms of sup norm). In contrast, for $\alpha=0.1$ we obtain distributions with few points of support and large point masses. 
\begin{figure}[!h]
	\centering
	\includegraphics[width=.9\linewidth]{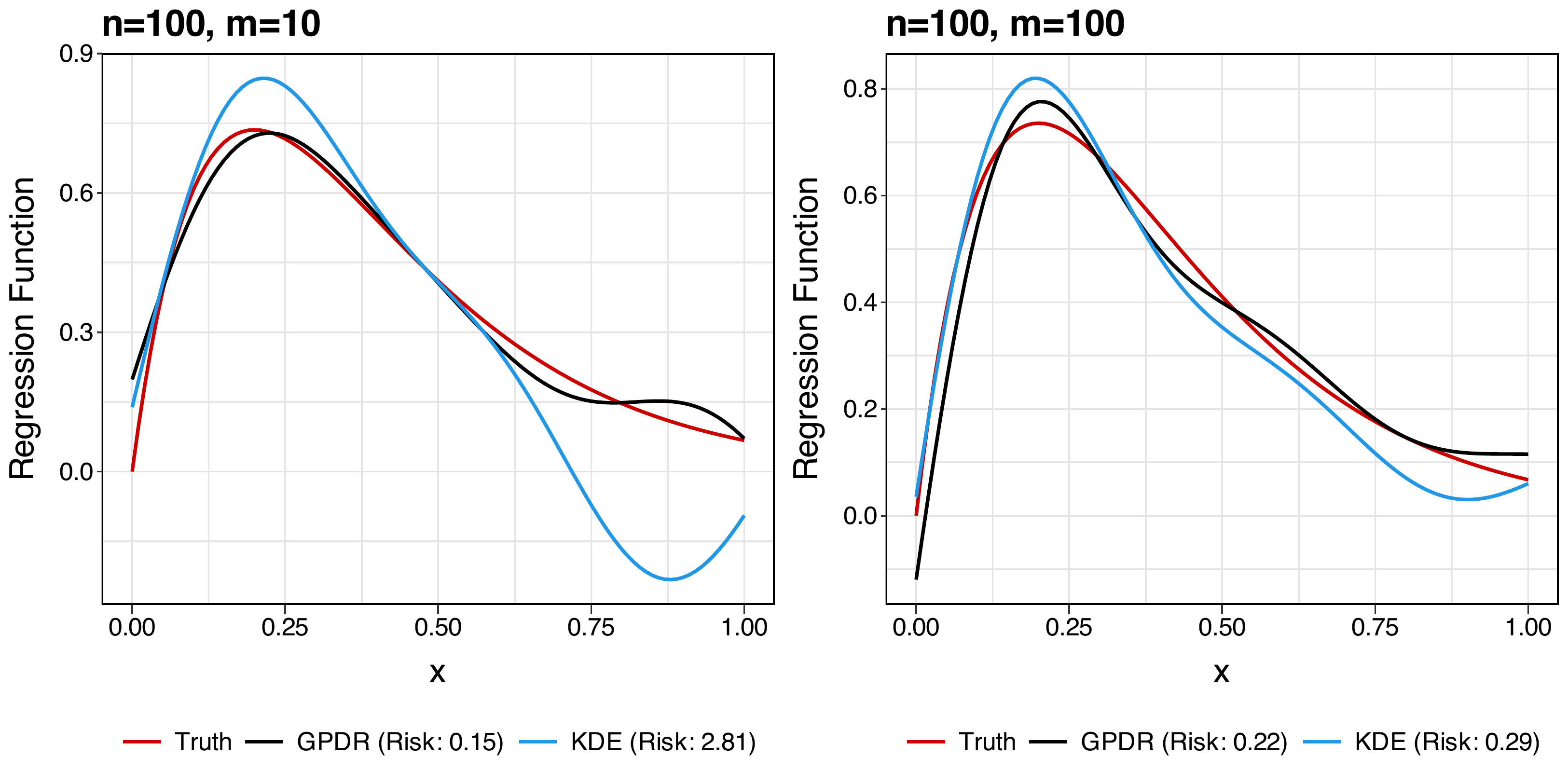}
	\caption{\label{fig:reg func DP 0.1}Given a simulated data with DP concentration parameter 0.1, this compares regression function estimates based on KDE and GPDR with the truth.}
\end{figure}
\begin{figure}[!h]
	\centering
	\includegraphics[width=\linewidth]{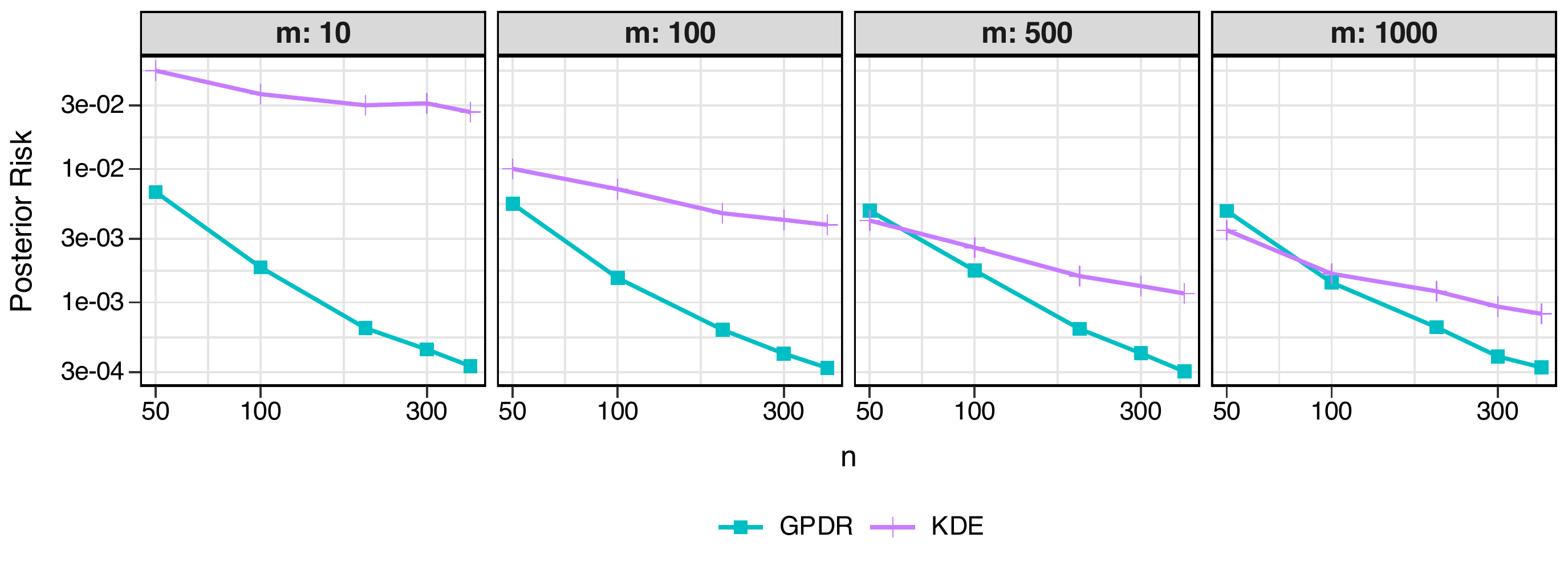}
	\caption{\label{fig:risk DP 0.1}The empirical posterior risk for different $n$ and $m$ under 100 runs with DP concentration parameter 0.1.}
\end{figure}
We now repeat the simulation experiments, using $\alpha=0.1$ instead of $\alpha=25$ in the data generation process (\ref{sim:ys}). 
Figure~\ref{fig:reg func DP 0.1} compares regression function estimates using GPDR and KDE with the truth for two values of $m$. Due to the discreteness, it is challenging to estimate densities $Z_i$, particularly for small $m$. The regression function estimation in KDE suffers from this, with estimates being far from $f_0$ at both ends. With GPDR not relying on density estimates, it provides a significantly better estimate even for such small $m$. As $m$ increases to 100, the KDE estimation improves, but GPDR still performs better in comparison.  Figure~\ref{fig:risk DP 0.1} compares their performances over $100$ independent simulation runs. It shows that GPDR performs significantly better than KDE across all settings. In particular, GPDR producing risk estimates that are over $100$-fold smaller than those of KDE for smaller $m$. This reflects the drastic gains afforded by GPDR over KDE when the true distributions are discrete and when the density estimation becomes challenging (small $m$). 

\begin{figure}[!h]
	\centering
	\includegraphics[width=\linewidth]{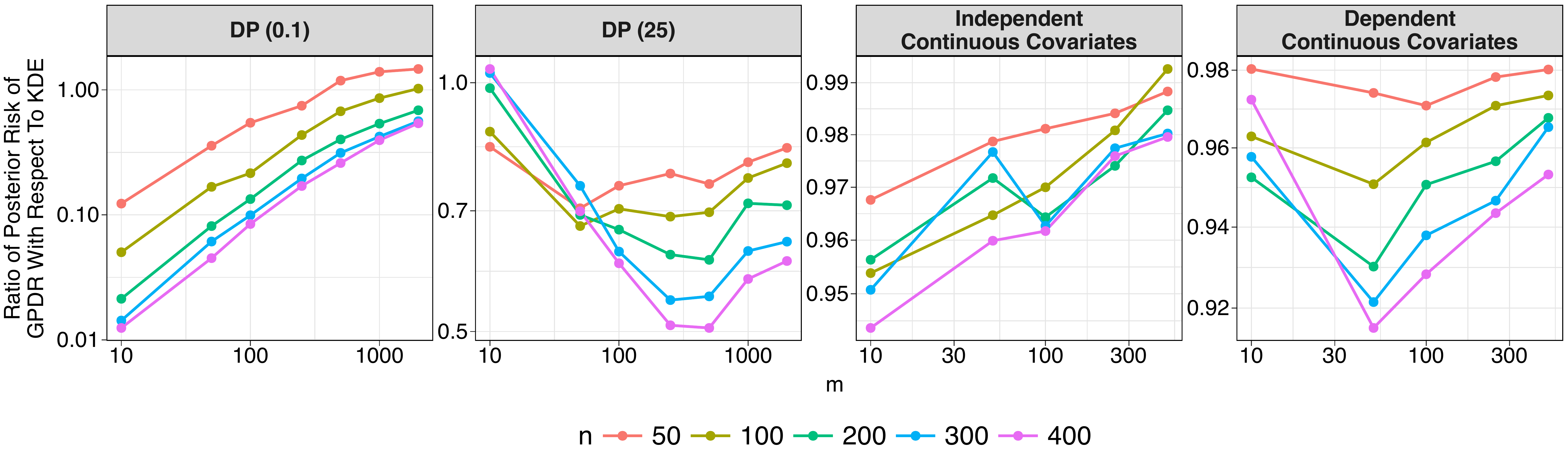}
	\caption{\label{fig:risk ratio}For different $n$ and $m$ with 100 simulation runs, this presents the ratio of empirical posterior risks of GPDR to KDE. $\rho$ is set to 0.5 for simulating dependent covariates.}
\end{figure}
Figure~\ref{fig:risk ratio} summarizes performance across all 4 scenarios considered in the simulation experiments (DP with $\alpha=25$, DP with $\alpha=0.1$, continuous distributions with i.i.d. samples, continuous distributions with correlated samples for $\rho=0.5$). We  plot the ratio of the posterior risk of GPDR to KDE. Combined across all settings, GPDR in general performs better than KDE. Only for small $n$ and large $m$, the two perform similarly. For small $m$ and highly discrete case $\alpha=0.1$, the gains from GPDR over KDE are dramatic. Also, with increasing $n$, the GPDR risk continues to improve over KDE, i.e., the risk ratio decreases across all 4 scenarios.

}

\subsection{\blue{Linear and non-linear GPDR}}\label{sec:simnl}

\blue{In this Section, we conduct a numerical experiment to compare the performance of linear  (\ref{gp_dr}) and non-linear (\ref{eq:nlgpdr}) GPDR. The non-linear GPDR does not enjoy the conjugacy of the posterior as the link function $\phi$ is also unknown, in addition to the regression function $f$, and the joint posterior of $(f,\phi)$ is not conjugate even if the kernel parameters are held fixed. So we use a MCMC-based implementation in STAN. For fairness of comparisons, the linear GPDR used for this experiment is also implemented in a fully Bayesian way in STAN (as detailed in Section \ref{sec:bayes}).

We consider both a linear and a non-linear generative model. For the linear case, the generative model is (\ref{eq:dgp}), which makes the linear GPDR analysis model (\ref{gp_dr}) correctly specified. Of course, the non-linear model (\ref{eq:nlgpdr}) is also a correctly specified in this case but constitutes a richer class of models. 
For the non-linear generative model we simulate $\mathbb E_{\varepsilon_i} y_i = \phi(\mathbb E_{Z_i} f_0)$ with a non-linear link $\phi(x) = \exp(x)$. In both setup the true regression function is $f_0(x) = \sin (2\pi x)$ for $\bx \in [0,1]$. We use $m=25$ and $n=20$ as the full Bayesian implementation scales poorly when $n$ or $m$ is moderately large. For non-linear GPDR we model the link as a polynomial (cubic) basis expansion with normal prior on the basis coefficients. 

Figure \ref{fig:lin} estimates of the regression function $f$ from linear and non-linear GPDR for the linear data generation scenario. We see that the linear model offer an accurate posterior mean estimate for the function $f_0$ with tight uncertainty bounds. The posterior mean of $f$ from the non-linear GPDR is also accurate, however, it has very large uncertainty bounds. This is expected as the non-linear GPDR is a much richer class of models and the uncertainty in estimation of the unknown link function is propagated into the uncertainty estimates for $f$. The linear model, being correctly specified and without redundant parametrization, is unsurprisingly more efficient. Thus if the data generation process can be well approximated by a linear model, it serves well to use the linear GPDR to obtain more confident estimates. 

\begin{figure}[h]
    \centering
    \includegraphics[scale=0.13]{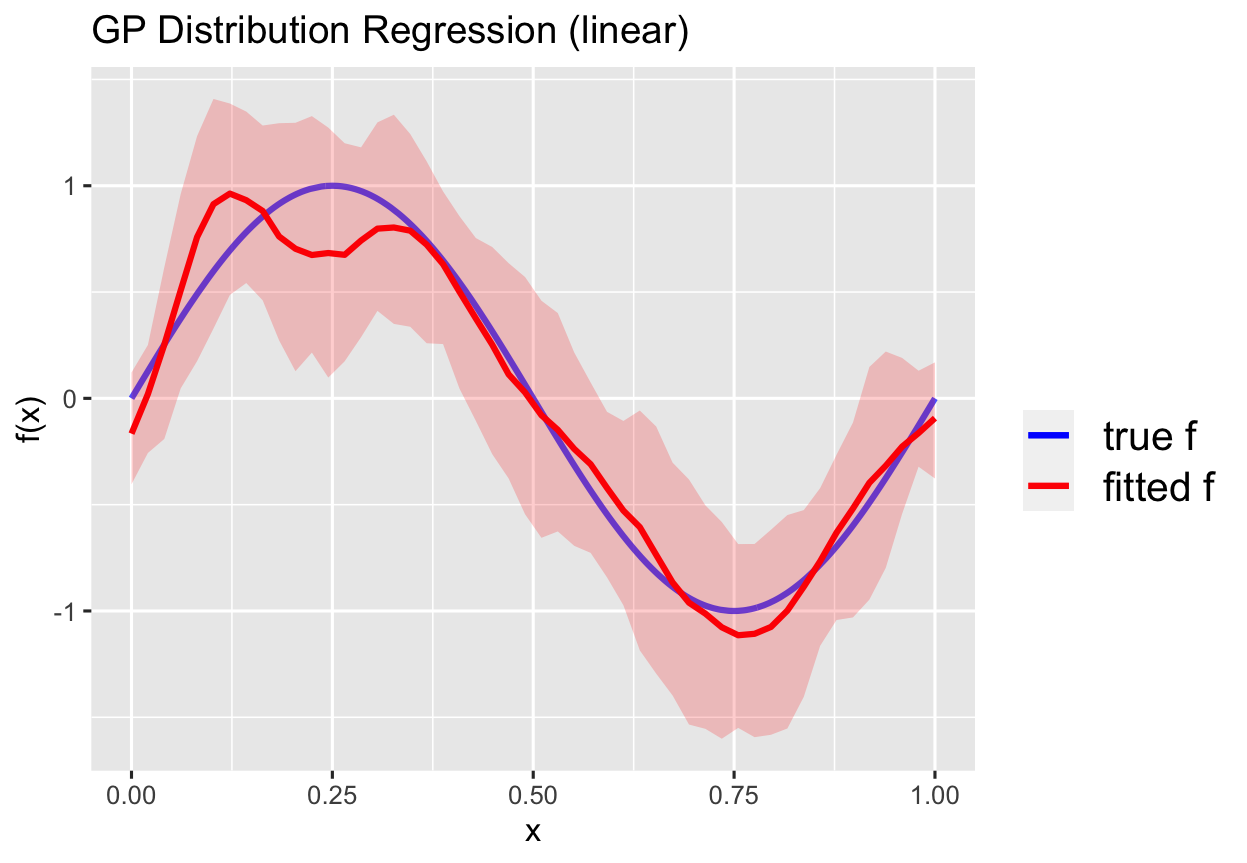}
    \includegraphics[scale=0.13]{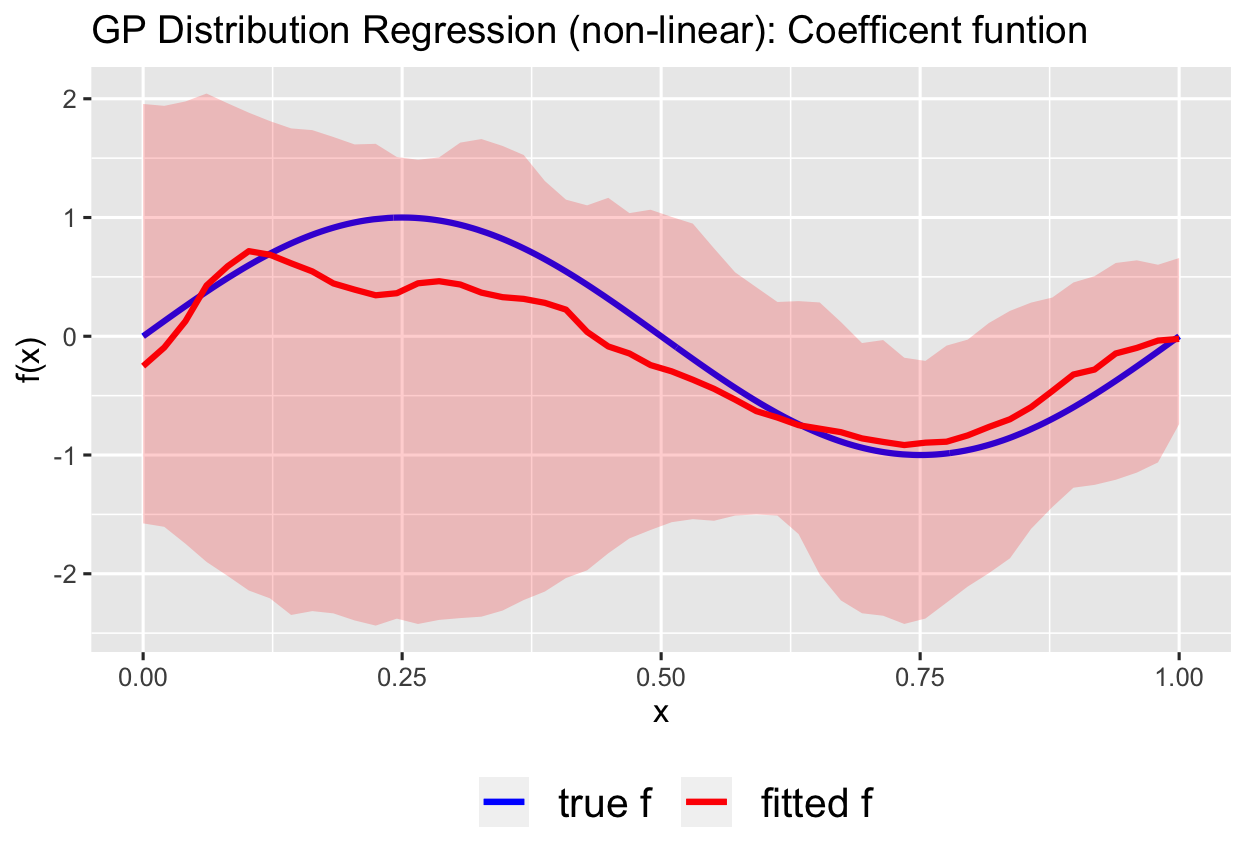}
    \caption{Comparison of estimates of the regression function $f_0$ from linear GPDR (left) and non-linear GPDR (right) when the true generative model is linear.}
    \label{fig:lin}
\end{figure}

Figure \ref{fig:lin} estimates of the regression function $f_0$ from linear and non-linear GPDR for the non-linear data generation scenario. Due to the non-linearity of the link function, the linear model is now misspecified, yielding a biased estimate of $f_0$. The non-linear model once again offers an accurate point estimate, although the intervals are still wide. 

\begin{figure}[h]
\centering
\includegraphics[scale=0.13]{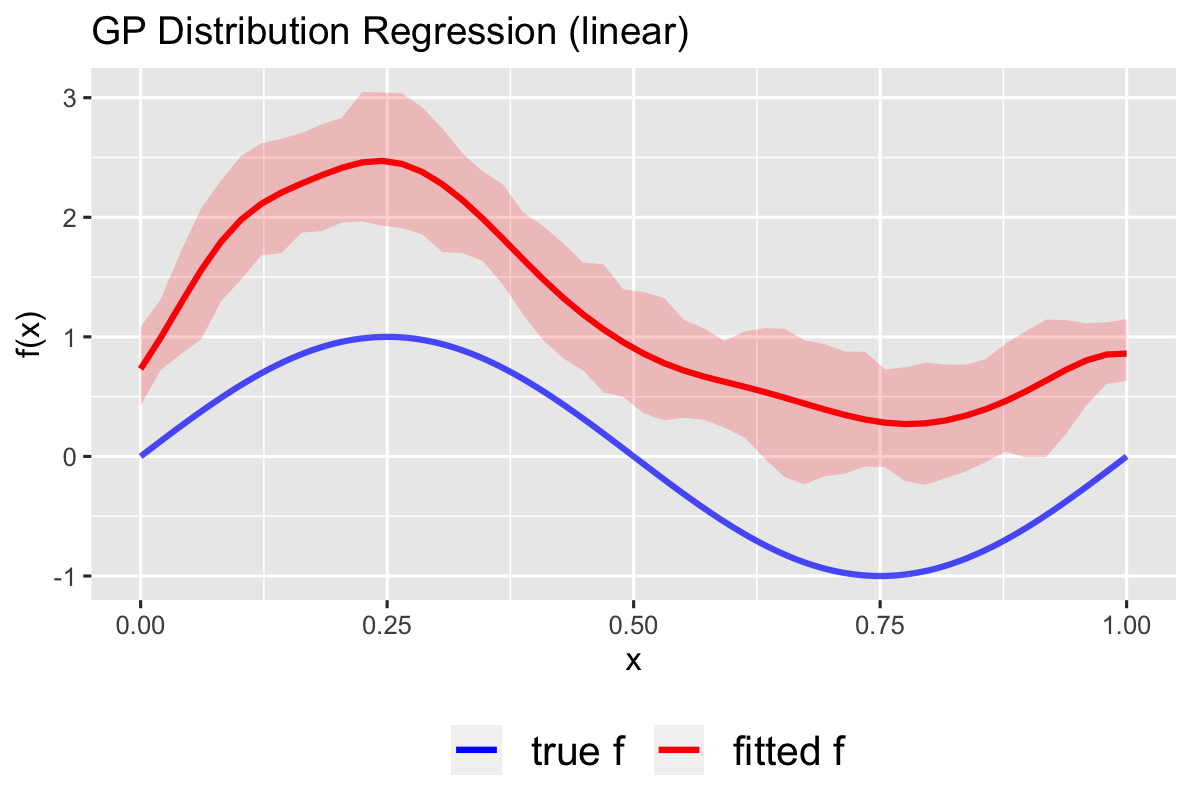}
\includegraphics[scale=0.13]{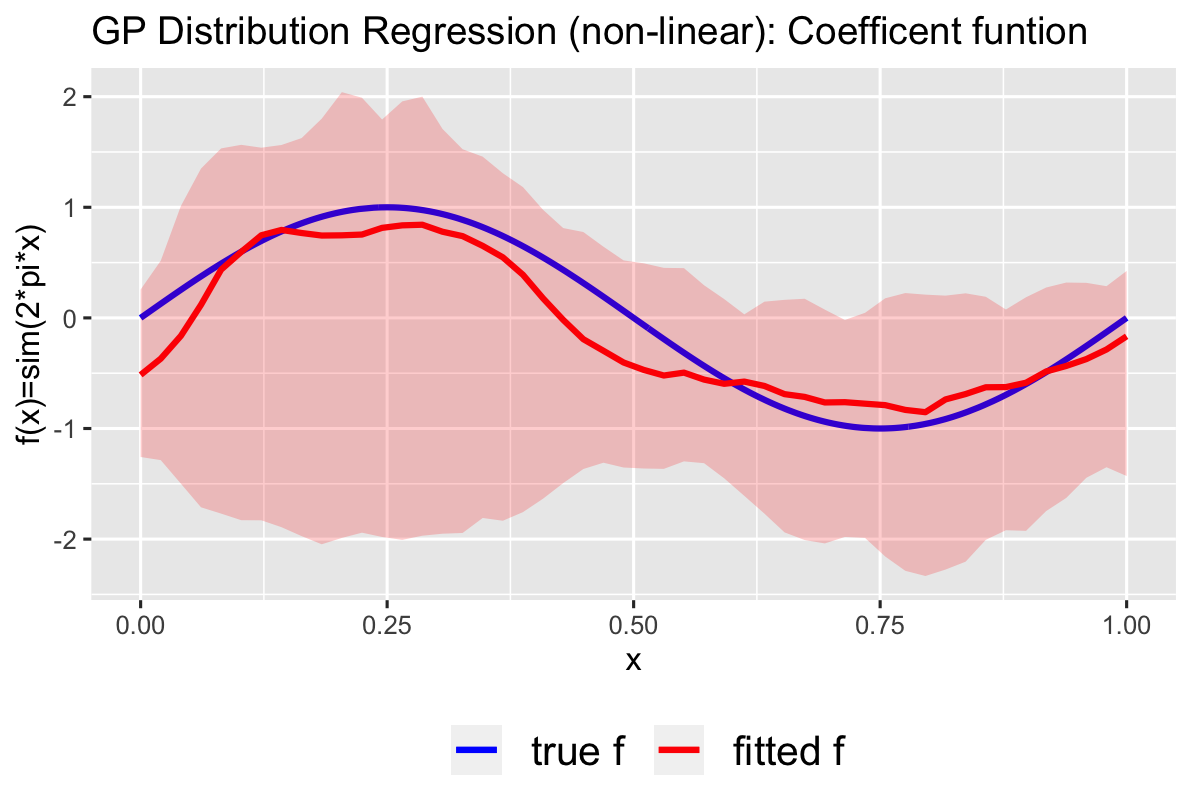}
\caption{Comparison of estimates of the regression function $f_0$ from linear GPDR (left) and non-linear GPDR (right) when the true generative model is non-linear.}
\label{fig:nl}
\end{figure}

Finally, Figure \ref{fig:link} presents the estimates of the link function $\phi$ from the non-linear GPDR for both the linear and non-linear data generation scenarios. We see that non-linear GPDR estimates the link accurately in both the linear scenario where $\phi(x)=x$ and in the non-linear scenario where $\phi(x)=\exp(x)$.

\begin{figure}
\centering
\includegraphics[scale=0.13]{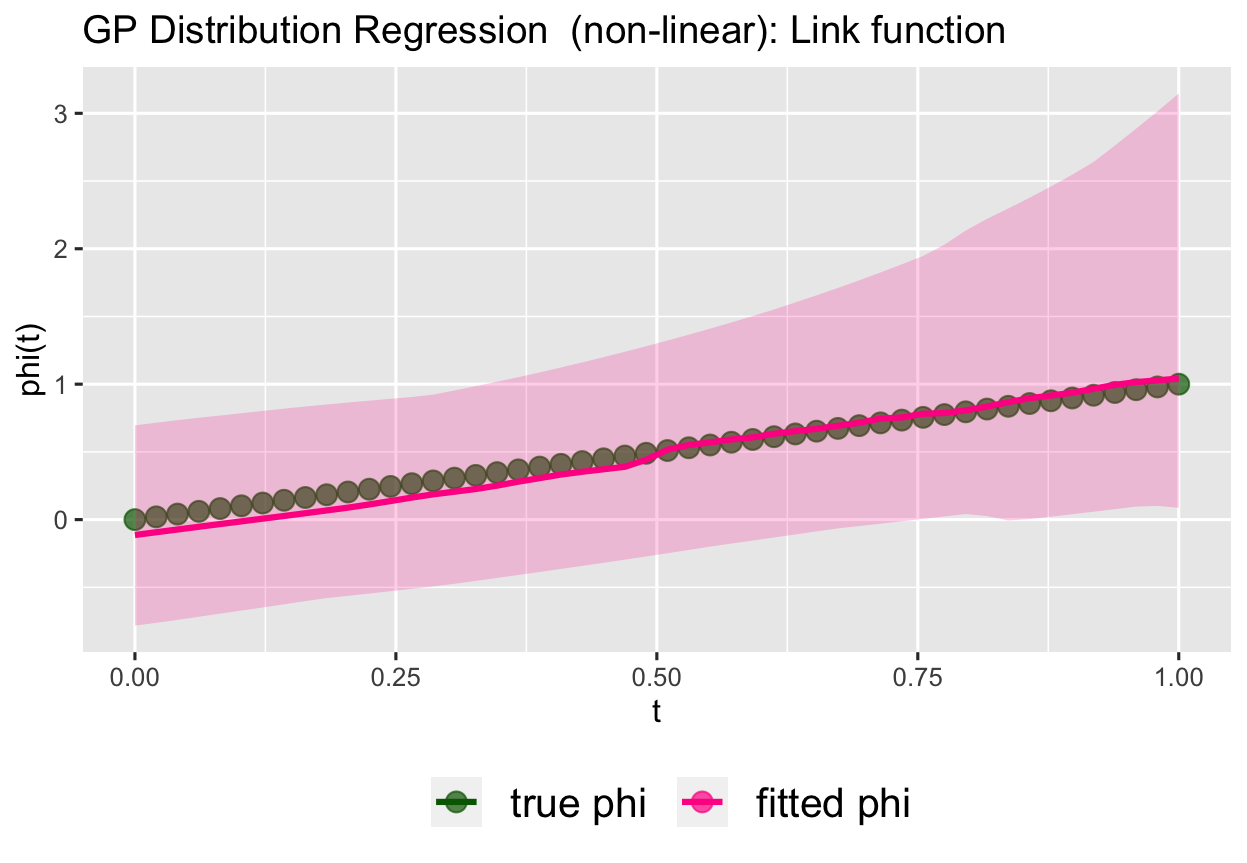}
\includegraphics[scale=0.13]{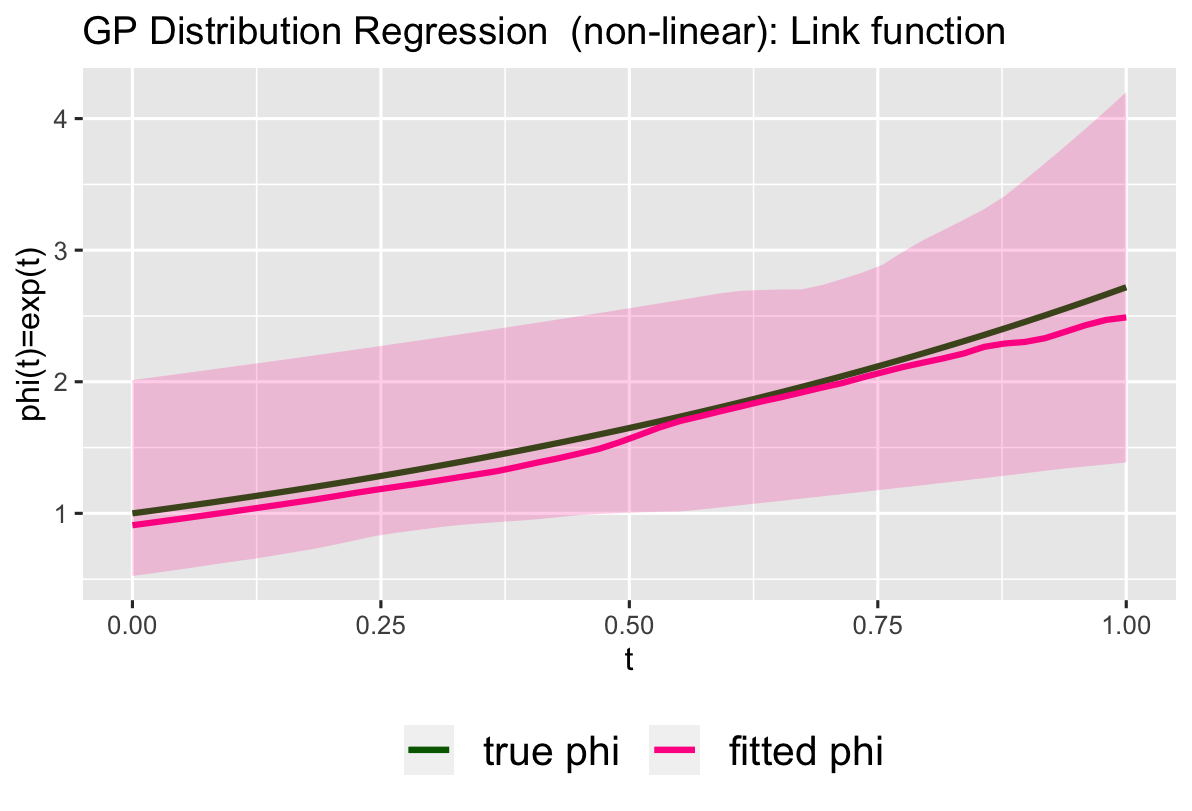}
\caption{Comparison of estimates of the link function $\phi$ from non-linear GPDR when the true generative model is linear (left) and non-linear (right).}
\label{fig:link}
\end{figure}

Overall, this experiment presents a proof-of-concept of utility of the non-linear GPDR. The results reveal the usual tradeoffs of using a simple and a richer model. The simple linear model, when correctly specified, offers more efficient estiamtion but can lead to highly biased inference under gross misspecification. The richer non-linear model performs well under various data generation scenarios but this robustness comes at the expense of precision due to the considerably higher dimensionality of the model. 
}
\section{NHANES activity data analysis}
The National Health and Nutrition Examination Survey (NHANES) is a long-standing survey study conducted by the Centers for Disease Control and Prevention and the National Center for Health Statistics to study the nutrition and health of the US population. Relatively recent samples (2003, 2005 cycles) collected physical activity data via hip worn accelerometers.  The data consist of minute resolution records of activity intensity beginning from 12:01 a.m. the day after participants’ health examination. Subjects were instructed to wear the accelerometer for 7 consecutive days during the daytime,  except when showering or bathing. The data used herein was processed and curated in the {\em rnhanes} package, as described in \cite{leroux2019organizing}.  

The database contains relatively high frequency activity data recorded from accelerometers. These are often summarized by a single measure of daily activity \citep[see the discussion in][]{bai2016activity}.
However, it is also of interest to consider the functional shape of the activity as a time ordered function \citep{goldsmith2016new}. We further argue that it is also of interest to investigate time-invariant, i.e., distributional aspects of the data, to summarize the role of activity, whenever it occurred. 
Therefore, we consider \blue{our distribution regression approach for this problem treating} the activity counts as \blue{order-invariant} covariates, as if repeated measures, ignoring the time-ordering of data points and assuming them to be samples from a subject-specific distribution. We illustrate that in this setting we can have superior performance compared to typical functional data analysis which considers the time \blue{(order)} and not the distribution as informative.

The data comprises 2,719 subjects recorded in 2003-2004 \& 2005-2006 NHANES survey cycles with activity count data for every minute during one week (i.e., there are $1,440 = 60$ minutes/hour $\times 24$ hours) points for every subject every day. In this analysis, we use it to predict the age of subjects reported during the survey. Activity has been shown to be correlated with age, with \cite{schrack2014assessing} reporting an expected 1.3\% decline in total activity per year.
For simplicity, we use only the activity data on Friday, because empirically mean activity on this day had the highest correlation with age. Care should be taken in interpreting our results, since \cite{smirnova2020predictive} showed a high correlation between accelerometry estimated activity, in the NHANES dataset, and all cause mortality.

Denote the age of subject $i$ as $y_i$ and the log-transformed activity counts as $\bx_{ij} = \log(1 + count_{ij})$ for every minute $j$. Consider two situations: 1) $\bx_{ij}$ as functional data with equally spaced time indices, $j$, via functional linear regression (FLR), 2) ignoring that $j$ corresponds to the $j^{th}$ time-point and treating $\{\bx_{ij}\}$ as samples from some distribution, $Z_i$, and utilizing distribution regression approaches. We will use both our Gaussian process distribution regression (GPDR) with low rank approximation as described in Section \ref{sec:low_rank} as well as the  \textit{Bayesian Linear Regression} (BLR)  
method introduced in \cite{law2018bayesian}. 
The BLR model was chosen as it is a simplified version of the BDR model used in Section \ref{sec:sim}, which is computationally intensive. Of note, only the uncertainty in the regression function is modeled, while sampling uncertainty is not. 

For functional analysis, we use penalized functional linear regression and smooth the curve using a 10 cubic spline basis with equally spaced knots. The penalization parameter is estimated from generalized cross validation. In the distribution regression situation, we use the low rank approximation of Section \ref{sec:low_rank} and set $k=10$ fixed. We further use a Mat\'ern kernel with regularity 2.5 (thus is twice differentiable). 
As suggested in \cite{kammann2003geoadditive}, we set the scale parameter as the range $\max \bx_{ij} - \min \bx_{ij}$. Other parameters were fit from generalized cross validation. We use the same kernel for BLR method and set 10 and 50 evenly spaced landmark points in the range of $[\min \bx_{ij}, \max \bx_{ij}]$. All other hyperparameters were set as the default values.

Each algorithm is evaluated using 5-fold cross validation. The mean R squared and its 95 \% confidence interval are estimated from 500 independent runs. In each run the R squared in the validation data and record the average value across 5 folds. Results are shown in Table \ref{tb:cv}. The density regression method significantly outperforms the default functional regression method and does slightly better than the BLR method in \cite{law2018bayesian}. There is a roughly 26\% improvement in out-of-sample R squared compared to functional linear regression and 6\% improvements compared to the Bayesian linear regression.  

\begin{table}[H]
	\caption{5-fold cross validation results for different algorithms. FLR means functional linear model that respect the time order of the data. GPDR is our Gaussian process distribution regression using low rank approximation method which ignoring any time information. BLR is the Bayesian linear regression model of  \cite{law2018bayesian} with 10 or 50 landmark points.}
	\centering
	\begin{tabular}{|c|c|c|c|c|}
		\hline
		& FLR & GPDR & BLR:10 & BLR:50 \\
		\hline
		R squared & 14.7\% & 18.5 \% & 17.6 \% & 17.6 \%\\
		\hline
		95\% CI & (14.2 \%, 15.0\%) & (17.9 \%, 18.8\%) & (17.2 \%, 17.8 \%) & (17.2 \%, 17.8 \%) \\
		\hline
	\end{tabular}
	\label{tb:cv}
\end{table}

The results demonstrate that distribution regression is often worth considering, even if the data is functional in nature and not random samples from individual distributions. This clearly can happen when order invariant summaries of the data, rather than ordered functional summaries, are more related to the outcome.



\section{Discussion}
We present a simple approach to regression using distribution-valued (repeated measures) covariate. We frame the model as a generalization of the typical Bayesian GP regression. This has several advantages. The method becomes simply a GP regression with the kernel averaged across repeated measures. We can obtain inference from the model directly using the observed covariate samples without having to estimate the subject-specific distributions or densities. All advantages of GP regression are retained, such as exploiting conjugacy in a Gibbs sampler and using low-rank GP approximations. The hierarchical formulation allows incorporating other subject-specific information like additional scalar-valued covariates and longitudinal or clustered data structures. Future work will explore these extensions in the context of specific applications. 

Theoretically, we present a comprehensive, and to our knowledge, the first set of results for Bayesian distribution regression using GP. We show our method has the same optimal finite-sample error bounds of function estimation as the typical GP regression does, under the same assumptions on the function smoothness and kernel choice as long as the class of subject-specific distributions is separable, i.e., rich enough to identify the regression functions. \blue{The theory can even account for  dependence among samples within a subject, possibly occurring due to time of measurement for the samples. However, it does not consider possible dependence among subjects. Such dependence will be less common than within-subject dependence, especially if the model accounts for other subject level covariates as in Section \ref{sec:subject}. However, we conjecture that similar theoretical results will hold for this setting as well as long as the dependence across subjects are controlled. Studying this will be an important future direction. 

In our approach we do not attempt to model or infer on the true covariate distributions $Z_i$, which we consider to be nuisance parameters. This is deliberate and a strength of the approach, as we work only with the observed samples which enables implementing our model as standard GP regression.  We show theoretically and empirically that we can recover the true regression function despite the misspecification on account of only using the repeated measures. However, our Bayesian hierarchical formulation is amenable to modeling  the true distributions, if desired. One would then need to add another hierarchy to the analysis model (\ref{gp_dr}) modeling the samples $\bx_{ij}$ as draws from some parametric or non-parametric distribution family. This would, however, lead to both computational and theoretical challenges as the posterior for the regression function will no longer be available in closed form. Addressing these would be another potential direction for future research, as would be extending the theory to the non-linear GPDR presented in Section \ref{sec:nonlin}.

There is further room to improve scalability of the approach. Despite the approximations (low-rank GP, conjugate posteriors) used in the current implementation, the approach cannot be deployed for large $n$ or $m$. In the future we plan to use state-of-the-art GP approximations to further improve scalability and develop a publicly-available software for GPDR.}

\section*{Acknowledgemnt}
This work was partially supported by the following grants: NIEHS R01 ES033739, NIBIB R01 EB029977, NIBIB P41 EB031771, NIDA U54 DA049110, NINDS R01 NS060910, NIMH R01 MH126970.

\vskip 0.2in
\bibliographystyle{biometrika}
\bibliography{reference}

\newpage

\begin{center} \Large \textbf{Supplementary Materials}
\end{center}

\setcounter{section}{0}
\renewcommand\thesection{S\arabic{section}}
\renewcommand\theequation{S\arabic{equation}}
\renewcommand\thefigure{S\arabic{figure}}
\renewcommand\thetable{S\arabic{table}}

\section{\blue{Examples of weak separability}}\label{sec:weak}

\blue{As mentioned in Section \ref{note_and_assump} weak separability will be satisfied for most distributional processes with support on a sufficiently rich collection of distributions $\calS$. We give a few examples:\\

\noindent \textbf{Dirac distributions:} Let $\calS=\{\delta_\bx \mid \bx \in \mathcal X\}$. So any subset of $A$ of $\calS$ corresponds to a set $B(A) \subseteq X$ such that $\{A= \delta_\bx \mid \bx \in B(A) \}$. 
Then a distributional process $\mathcal Z$ with support on $\calS$ corresponds to a distribution $\nu(\calZ)$ with support on $\mathcal X$ such that for any $A \subseteq \calS$ we have $P_{\mathcal Z}(A) = P_{\nu(\calZ)}(B(A))$.

Now suppose we have a set $A \in \calS$ such that $P_\calZ(A)=1$ and for every distribution $Z \in A$, we have $\mathbb E_Z(f_1) = \mathbb E_Z(f_2)$ for two functions $f_1,f_2$ in a functional class $\calF \subseteq \mathcal C_0$, the set of all continuous functions on $\mathcal X = [0,1]^d$. Then we will show that $f_1=f_2$. 

Note that $P_\calZ(A)=1$ implies $P_{\nu(\calZ)}(B(A))=1$. So for any $\bx \in \calX$, there exists a sequence $\bx_1,\bx_2,\ldots$ in $B(A)$ which converges to $\bx$. So the sequence $\delta(\bx_1),\delta(\bx_2),\ldots,$ is in $A$ and we have $\mathbb E_{\delta_{\bx_k}}(f_1)=E_{\delta_{\bx_k}}(f_2)$ for all $k$, i.e., $f_1(\bx_k)=f_2(\bx_k)$  for all $k$. Taking the limit $k \to \infty$ and using continuity of $f_1,f_2$ we have $f_1(\bx)=f_2(\bx)$. This holds for all $\bx \in \calX$ proving weak separability. 
\\

\textbf{Continuous distributions:} Consider $d=1$, and $\calS$ denote the collection of all $Beta(a,b)$ distributions on $\calX=[0,1]$. Let $\Theta=(0,\infty) \times (0,\infty)$ denote the set of all valid parameter choices for the Beta distribution. Any subset of $A$ of $\calS$ corresponds to a set $B(A) \subseteq \Theta$ such that $\{A= Beta(a,b) \mid \theta=(a,b) \in B(A) \}$. Then a distributional process $\calZ$ on $\calS$ corresponds to a distribution $\nu(\calZ)$ on $\Theta$ such that for any $A \subseteq \calS$ we have $P_{\mathcal Z}(A) = P_{\nu(\calZ)}(\theta \in B(A))$.

We will show that if $\calZ$ is supported on all of $\calS$, i.e., it gives positive probability around any Beta distribution then $\calZ$ is weakly separable. 

For $\theta=(a,b)$, we write $Beta(a,b)=Beta(\theta)$. Suppose we have a set $A \in \calS$ such that $P_\calZ(A)=1$ and for every distribution $Z \in A$, we have $\mathbb E_Z(f_1) = \mathbb E_Z(f_2)$ for two functions $f_1,f_2$ in a functional class $\calF \subseteq \mathcal C_0$. As $\calZ$ is supported on all of $\calS$ we have equivalently that $\nu(\calZ)$ is supported on entire $\Theta$ and $P_\calZ(A)=1$ implies $P_{\nu(\calZ)}(B(A))=1$. For any $\theta \in \Theta$, there exists a sequence $\theta_1,\theta_2,\ldots$ in $B(A)$ which converges to $\theta$. Then as $\mathbb E_{Beta(\theta_k)} f_1 = \mathbb E_{Beta(\theta_k)} f_2$ for all $k$, Beta$(\theta_k) \overset{d}{\implies} \mbox{Beta}(\theta)$, and $f_1, f_2$ are bounded functions (continuous on compact domain), taking limit of $k \to \infty$ and applying the Portmanteau theorem, we have $\mathbb E_{Beta(\theta)} f_1 = \mathbb E_{Beta(\theta)} f_2$. This holds for all $\theta \in \Theta$. For any $x$ in $(0,1)$, consider $\theta = (rx,r(1-x))$. Then $Beta(\theta) \overset{d}{\implies} \delta_x$ as $r \to \infty$. Applying the Portmanteau theorem once again, we have $\mathbb E_{\delta_x} f_1 = \mathbb E_{\delta_x} f_2$ for all $x \in (0,1)$ which, by continuity, implies $f_1(x) = f_2(x)$ for all $x \in [0,1]$.

Similar logic will hold for many other rich collection $\calS$ of continuous distributions as long as $\calZ$ is supported on all of $\calS$ (e.g., normal distributions or normal mixtures) truncated to transformed to be on $[0,1]^d$.}

\section{\blue{Non-universality of the linear kernel}}\label{sec:nonuni}

\blue{We discussed in Section \ref{sec:r0} how the GP distribution regression (\ref{eq:gpan2}) can be expressed as a standard GP regression (\ref{eq:gpan2dist}) directly with distributional inputs $Z_i$ (or $\widehat Z_i$, in practice) and the induced kernel $\mathbb K$. Unsurprisingly, this kernel will not be universal as it is linear in the kernel mean embeddings of the distributions (\ref{eq:linearkme}). We give a concrete counter-example here demonstrating this non-universality. 

Let $\calX=[0,1]$ and $\calS$ denote any collection of distributions that contain all Dirac (degenerate) distributions on $\calX$ and atleast some non-degenerate distributions. Then there exists some distribution $Z_0 \in \calS$ and some $\delta > 0$ such that $\mathbb E Z_0^2 > (\mathbb E Z_0)^2 + \delta$. 

Let $f \sim GP(0,K)$ and define $F(Z) = \mathbb E_Z f$. Then as discussed in Section \ref{sec:banach}, $F$ is a Gaussian random element with kernel $\mathbb K$ and taking values in a separable Banach space $\mathcal B$. So the RKHS $\mathcal H_{\mathbb K}$ lies within the support of $F$ (Lemma 5.1 of \cite{van2008reproducing}). 

Now let's consider the functional $F_0(Z)= (\mathbb E Z)^2$ for all $Z \in \calS$. We will show that $F_0$ cannot be approximated by the support of $F$ and hence by the RKHS $\mathbb K$. By the definition of $F$, any realization $F(\omega)(Z) = \mathbb E_{\bx \sim Z} f(\omega)(\bx)$ for all $Z \in \calS$ where $f(\omega)$ is a realization from $GP(0,K)$. 

Suppose there exists a realization $f(\omega)$ such that $F(\omega)(Z) = \mathbb E_{\bx \sim Z} f(\omega)(\bx)$ satisfies $\sup_{Z \in \calS} |F_0(Z) - F(\omega)(Z)| \leq \delta/2$. 

Choosing $Z = \delta_{\bx}$ we have $F(w)(Z) = f(\omega)(\bx)$ and $F_0(Z) = \bx^2$ implying
$\bx^2 - \delta/2 <  F(\omega)(\bx) <  \bx^2 + \delta/2$. This holds for all $\bx \in \calX$.  

Then we have $F(\omega)(Z_0) = \mathbb E_{\bx \sim Z_0} f(\omega)(\bx)$ satisfying
\[ \int (\bx^2 - \delta/2)Z_0(d\bx)  < \mathbb E_{\bx \sim Z_0} f(\omega)(\bx) < \int (\bx^2 + \delta/2)Z_0(d\bx), \mbox{i.e.,} \]
\[ \mathbb E Z_0^2 - \delta/2  < \mathbb E_{\bx \sim Z_0} f(\omega)(\bx) < \mathbb E Z_0^2 + \delta/2 \]
So we have $|F(\omega)(Z_0) - \mathbb E Z_0^2| \leq \delta/2$ implying $|F(\omega)(Z_0) - (\mathbb E Z_0)^2| \geq \delta/2$, i.e.,  
$|F(\omega)(Z_0) - F_0 Z_0| \geq \delta/2$ (as $\mathbb E Z_0^2 > (\mathbb E Z_0)^2 + \delta$). This contradicts the approximating property of $F(\omega)(Z_0)$. 

So the support of the Gaussian process $F \sim GP(0,\mathbb K)$ cannot approximate all functions arbitrarily closely and thus neither can $\mathcal H_{\mathbb K}$. So $\mathbb K$ is not universal. 
}

\section{Proofs}\label{sec:proofs}
\subsection{Definitions}\label{sec:def}
In this manuscript, we will mainly study the \textit{H\"{o}lder space} $C^{\alpha}[0, 1]^d$ for $\alpha > 0$. When writing $\alpha = k + \eta$, \blue{where $k = \floor{\alpha}$,} the space $C^\alpha[0,1]^d$ is the space of all functions supported in $[0, 1]^d$, whose partial derivatives of orders $(l_1,\cdots,l_d)$ exist for all non-negative integers $l_1,\cdots,l_d$ such that $l_1 + \cdots + l_d \le k$ and for which the highest order partial derivatives are H\"{o}lder continuous with order $\eta$ ($f$ being H\"{o}lder continuous with order $\eta$ if $|f(x) - f(y)| \le C\norm{x-y}^\eta$ for all $x, y$ and some constant $C$) .

Another functional space we will study is the \textit{Sobolev space} $H^\alpha[0,1]^d$ which contains all $[0, 1]^d \rightarrow \mathbb{R}$ functions $f$ such that
\begin{equation*}
	\int_{\mathbb{R}^d} \left(1 + \norm{\lambda}^2\right)^\alpha \left|\hat{f}(\lambda)\right|^2 d\lambda < \infty
\end{equation*}
where $\hat{f}$ is the Fourier transformation of $f$: $\hat{f}(\lambda) = (2\pi)^{-d}\int \exp(i\lambda^T t) f(t) d t$. 

An order $\alpha$ Mat\'ern kernel for $d$ dimensional process has the form:
\begin{equation*}
	K(s, t) = \int_{\mathbb{R}^d} \frac{e^{i\lambda^T (s-t)}}{(1 + \norm{\lambda}^2)^{\alpha + d / 2}} d \lambda.
\end{equation*}

\subsection{Strong separability of Dirichlet Process}
\begin{proof}[Proof of Lemma \ref{lem:dp}]
	Using stick breaking representation, we know that the sample probability mass function $p(x)$ has the form
	\begin{equation} \label{eq:lem1}
	p(x) = \sum_{k=1}^{\infty} \beta_k \cdot \delta_{\bx_k}(x),
	\end{equation}
	where $\beta_k = \beta'_k \prod_{i=1}^{k-1} (1 - \beta_i')$ for $\beta_k'$ i.i.d follows $\text{Beta}(1, \alpha)$ and $\bx_k$ i.i.d follows $\blue{\upsilon}$ with $\delta_\bx$ the point mass at $\bx$. Also, it is clear that the mean measure for $\text{DP}(\blue{\upsilon}, \alpha)$ is just $\blue{\upsilon}$. Therefore, for bounded $f$, we can directly calculate $\mathbb{E}_{z\sim \text{DP}(\blue{\upsilon}, \alpha)} [(\mathbb{E}_z f)^2]$. We have:
	\begin{align*}
		\mathbb{E}_{z\sim \text{DP}(\blue{\upsilon}, \alpha)} [(\mathbb{E}_z f)^2] &= \mathbb{E}_{z\sim \text{DP}(\blue{\upsilon}, \alpha)} \left[ \left( \sum_{k=1}^\infty \beta_k f(\bx_k) \right)^2 \right] \\
		&= \mathbb{E}_{z\sim \text{DP}(\blue{\upsilon}, \alpha)} \left[ \sum_{ij}^\infty f(\bx_i) f(\bx_j) \beta_i\beta_j\right] \\
		&= \mathbb{E}_\blue{\upsilon} f^2 \cdot \sum_i \frac{2}{\alpha(1+\alpha)}\left(\frac{\alpha}{2+\alpha}\right)^i + (\mathbb{E}_\blue{\upsilon} f)^2 \sum_{i \neq j} \mathbb{E}[\beta_i\beta_j] \\
		&\ge  \mathbb{E}_\blue{\upsilon} f^2 \cdot \sum_i \frac{2}{\alpha(1+\alpha)}\left(\frac{\alpha}{2+\alpha}\right)^i = \frac{1}{1 + \alpha} \mathbb{E}_\blue{\upsilon} f^2
	\end{align*}
	Therefore, $\text{DP}(\blue{\upsilon}, \alpha)$ strongly separates the bounded functions with constant $C = 1 + \alpha$.
\end{proof}

\subsection{\blue{Distance between mean embeddings}}\label{sec:emb}
In the following sections, we denote $K_s$ as the function $K_s(t) = K(s, t)$. 

\begin{proof}[Proof of Lemma \ref{lem:kme}]
\blue{Note that $\mu_{\wzi} = \frac{1}{m} \sum_{j=1}^m K_{\bx_{ij}}$. So 
\begin{equation}\label{eq:cross}
    \begin{split}
        \norm{\mu_{\widehat{Z}_i} - \mu_{Z_i}}_\mathcal{H}^2 =& \norm{ \frac 1m \sum_{j=1}^m (K_{\bx_{ij}} - \mu_{Z_i})}_\mathcal{H}^2 \\
        =& \frac 1{m^2} \sum_{j=1}^m \norm{K_{\bx_{ij}} - \mu_{Z_i}}_\mathcal{H}^2  \\
       & \qquad + \sum_{j \neq j'} <K_{\bx_{ij}} - \mu_{Z_i},K_{\bx_{ij'}} - \mu_{Z_i}>_{\mathcal H}.
    \end{split}
\end{equation}
We will now show that the expectation of the cross terms in (\ref{eq:cross}) is 0. Using the RKHS property and applying (\ref{eq:meta}), 
$$\mathbb E_{\bx_{ij} \sim Z_i} <K_{\bx_{ij}},K_{\bx_{ij'}}>_{\mathcal H} = \int_{\bx} \int_{\bx'} K(\bx,\bx') Z_i(d\bx) Z_i(d\bx') = \mathbb K(Z_i,Z_i). $$
By the linear representation of $\mathbb K$ in (\ref{eq:linearkme}), we immediately have $\|\mu_{Z_i}\|_{\mathcal H}^2 = \mathbb K(Z_i,Z_i)$. 
Also, 
\begin{align*}
\mathbb E_{\bx_{ij} \sim Z_i} <K_{\bx_{ij}},\mu_{Z_i}>_{\mathcal H} &= \mathbb E_{\bx_{ij} \sim Z_i} \int K(\bx_{ij},\bt) Z(d\bt) \\
&= \int_{\bx} \int_{\bt} K(\bx,\bt) Z_i(d\bx) Z_i(d\bt) = \mathbb K(Z_,Z_i).
\end{align*}
Thus all terms when expanding $<K_{\bx_{ij}} - \mu_{Z_i},K_{\bx_{ij'}} - \mu_{Z_i}>_{\mathcal H}$ has the same expectation, two of them are positive and two are negative,  and consequently 
$\mathbb E_{\bx_{ij} \sim Z_i} <K_{\bx_{ij}} - \mu_{Z_i},K_{\bx_{ij'}} - \mu_{Z_i}>_{\mathcal H}=0$. Finally, for any $Z$, $\|\mu_{Z}\|_{\mathcal H} \leq \kappa$ , the upper bound of the kernel $K$. Using the fact theat $\bx_{ij}$ are i.i.d., we then have:}

\begin{equation}\label{eq:kmebound}
\begin{split}
\mathbb{E} \norm{\mu_{\widehat{Z}_i} - \mu_{Z_i}}_\mathcal{H}^2 &= \mathbb{E}_{Z_i \sim \calZ} \mathbb{E}_{\bx_{ij} \sim Z_i} \norm{\frac{1}{m} \sum_j K_{\bx_{ij}} - \mu_{Z_i}}_\mathcal{H}^2 \\
	&= \mathbb{E}_{Z_i} \left[\frac{1}{m} \mathbb{E}_{\bx_{ij}} \norm{K_{\bx_{ij}} - \mu_{Z_i}}_\mathcal{H}^2 \right]\\
 &\blue{\leq \frac{1}{m} \mathbb{E}_{Z_i} \left[ 2 \mathbb{E}_{\bx_{ij}} \norm{K_{\bx_{ij}}}_\mathcal{H}^2 + 2\norm{\mu_{Z_i}}_\mathcal{H}^2 \right]}\\
 &\le \frac{4\kappa}{m}. 
 \end{split}
\end{equation}
\end{proof}

\subsection{\blue{Posterior distributions}}\label{sec:post}
\blue{We first express the posterior distributions of both GP distribution regressions (\ref{eq:gpan2}) and (\ref{gp_dr}) in terms of the kernel mean embeddings defined in (\ref{eq:kme}).}

We borrow similar notation as in \cite{szabo2016learning, caponnetto2007optimal}. 
Denote:
$$T_{{Z}} = \frac{1}{n} \sum_{i=1}^n \mu_{Z_i} \langle\mu_{Z_i}, \cdot\rangle, \hspace{2em} g_{{Z}} = \frac{1}{n} \sum_{i=1}^n y_i \mu_{Z_i}, \hspace{2em} \phi_{{Z}} = \frac{1}{n} \sum_{i=1}^n \mu_{Z_i} \otimes \mu_{Z_i}$$
where $(f_1 \otimes f_2) (s, t) = f_1(s)f_2(t)$. 
Naturally, $\mathcal{H} \otimes \mathcal{H}$ would be the closure of set $\{f_1 \otimes f_2: f_1\in \mathcal{H}, f_2 \in \mathcal{H}\}$, equipped with inner product as continuous extension of $\langle f_1 \otimes f_2, g_1 \otimes g_2 \rangle = \langle f_1, g_1 \rangle \langle f_2, g_2 \rangle$. Then we know that $g_{{Z}} \in \mathcal{H}$ and $T_{{Z}}: \mathcal{H} \rightarrow \mathcal{H}$ is Hermitian. Similarly, we can define $T_{\hat{{Z}}}, g_{\hat{{Z}}}, \phi_{\hat{{Z}}}$ by changing true distribution $Z_i$ to empirical distribution $\wzi$ based on the samples $\{\bx_{ij}\}$ 
, and they have the same properties. Following these notations, Gaussian process has the following posterior \citep{caponnetto2007optimal}:
\begin{align} \label{gp_post}
	f\ |\ \mathbb{Z}_n &\sim GP\left((T_{{Z}} + \sigma_n^2)^{-1}g_{{Z}},\ K - \left[(T_{{Z}} + \sigma_n^2)^{-1} \otimes \text{Id} \right] \phi_{{Z}} \right) \\
	f\ |\ \mathbb{D}_n &\sim GP\left((T_{\hat{{Z}}} + \sigma_n^2)^{-1}g_{\hat{{Z}}},\ K - \left[(T_{\hat{{Z}}} + \sigma_n^2)^{-1} \otimes \text{Id} \right] \phi_{\hat{{Z}}} \right)
\end{align}
where $\sigma_n^2 = \sigma^2 / n$ and Id is the identity operator. Here, the operator $T_1 \otimes T_2: \mathcal{H} \otimes \mathcal{H} \rightarrow \mathcal{H} \otimes \mathcal{H}$ is naturally defined through continuous extension of $(T_1 \otimes T_2) (f_1 \otimes f_2) = T_1(f_1) \otimes T_2(f_2)$. 
Denote 
$f_{{Z}} = (T_{{Z}} + \sigma_n^2)^{-1}g_{{Z}}$ and $f_{\hat{{Z}}} = (T_{\hat{{Z}}} + \sigma_n^2)^{-1}g_{\hat{{Z}}}$. Also $M_{{Z}} = \left[(T_{{Z}} + \sigma_n^2)^{-1} \otimes \text{Id} \right] \phi_{{Z}} $ and $M_{\hat{{Z}}} = \left[(T_{\hat{{Z}}} + \sigma_n^2)^{-1} \otimes \text{Id} \right] \phi_{\hat{{Z}}}$. Therefore $M_{{Z}}, M_{\hat{{Z}}} \in \mathcal{H} \otimes \mathcal{H}$. 

\blue{Note that the operators $T_Z$ ($T_{\widehat Z}$) will feature heavily in expression of the posterior mean and variance in (\ref{gp_post}). Before proceeding with the proof of the theorem, below we provide a bound on the distance $\|T_Z - T_{\widehat Z}\|$ in terms of the mean embedding distance (Lemma \ref{lem:kme}) that will be crucial in bounding $R_n^1$.

\begin{lemma}\label{lem:tz}
    $\mathbb E \norm{(T_{{Z}} - T_{\hat{{Z}}})}^2_{\mathcal{L(H)}} \leq 2\kappa \mathbb E \norm{\mu_{Z_i} - \mu_{\widehat{Z}_i}}_\mathcal{H}^2 \leq 8 \kappa^2 / m $ for operator norm $\norm{\cdot}_{\mathcal{L(H)}}$.
\end{lemma}
}

\begin{proof}

To bound $\norm{(T_{{Z}} - T_{\hat{{Z}}})}_{\mathcal{L(H)}}$, we have:
\begin{align*}
	\norm{T_{{Z}} - T_{\hat{{Z}}}}_\mathcal{L(H)}^2 &\le \frac{1}{n} \sum_{i=1}^n \norm{\mu_{Z_i} \langle\mu_{Z_i}, \cdot\rangle - \mu_{\widehat{Z}_i} \langle\mu_{\widehat{Z}_i}, \cdot\rangle}_\mathcal{L(H)}^2
\end{align*}
And for every $\norm{\mu_{Z_i} \langle\mu_{Z_i}, \cdot\rangle - \mu_{\widehat{Z}_i} \langle\mu_{\widehat{Z}_i}, \cdot\rangle}$, apply it to an arbitrary function $f$.
\begin{align*}
	\norm{\mu_{Z_i} \langle\mu_{Z_i}, f\rangle - \mu_{\widehat{Z}_i} \langle\mu_{\widehat{Z}_i}, f\rangle}_\mathcal{H}^2 &= \norm{(\mu_{Z_i} - \mu_{\widehat{Z}_i}) \langle\mu_{Z_i}, f\rangle + \mu_{\widehat{Z}_i} \langle\mu_{Z_i} - \mu_{\widehat{Z}_i}, f\rangle}_\mathcal{H}^2 \\
	&\le \norm{f}_\mathcal{H}^2 \left(\left(\norm{\mu_{Z_i}}_\mathcal{H}^2 + \norm{\mu_{\widehat{Z}_i}}_\mathcal{H}^2\right) \norm{\mu_{Z_i} - \mu_{\widehat{Z}_i}}_\mathcal{H}^2\right)\\
	&\le 2\kappa \norm{f}_\mathcal{H}^2 \norm{\mu_{Z_i} - \mu_{\widehat{Z}_i}}_\mathcal{H}^2
\end{align*}
Therefore 
\begin{equation}\label{eq:tz}
    \norm{T_{{Z}} - T_{\hat{{Z}}}}_\mathcal{L(H)}^2 \le \frac{2\kappa}{n}\sum_i \norm{\mu_{Z_i} - \mu_{\widehat{Z}_i}}_\mathcal{H}^2,
\end{equation}
\blue{and the first inequality follows by taking expectation. The second inequality follows from Lemma \ref{lem:kme}.}
\end{proof}

\subsection[Bound R1 for L2 norm]{Bound for $R_n^1$ \blue{under random design (Theorem \ref{th:random})}}

\subsubsection{\blue{Decomposition of }
 $R_n^1$} 
 \label{sec:morder}
A naive bound for $R_n^1$ is to directly calculate it out:
\begin{align*}
	R_n^1 &= \mathbb{E}_{f_0} \int ||f - f_0||^2 (d \Pi_n(f|\mathbb{D}_n) - d \Pi_n(f|\mathbb{Z}_n)) \\
	&= \mathbb{E}_{f_0} \int_{[0,1]^d} \left( \mathbb{E}_{f|\mathbb{D}_n}(f-f_0)^2(s) - \mathbb{E}_{f|\mathbb{Z}_n}(f-f_0)^2(s) \right) ds \\
	&= \mathbb{E}_{f_0} \int_{[0,1]^d} \left( \mathbb{E}_{f|\mathbb{D}_n}(f^2(s)-2f(s)f_0(s)) - \mathbb{E}_{f|\mathbb{Z}_n}(f^2(s)-2f(s)f_0(s)) \right) ds \\
	&= \mathbb{E}_{f_0} \int_{[0,1]^d} \left[\text{Var}(f(s)|\mathbb{D}_n) - \text{Var}(f(s)|\mathbb{Z}_n) + \mathbb{E}^2(f(s) | \mathbb{D}_n) - \mathbb{E}^2(f(s) | \mathbb{Z}_n) \right.\\
	&\hspace{5em} - \left.2f_0(s)\left( \mathbb{E}(f(s) | \mathbb{D}_n) - \mathbb{E}(f(s) | \mathbb{Z}_n) \right) \right] ds.
 \end{align*}
\blue{This leads to (\ref{eq:decomp}).} All exchange of integral and expectation above should be legal because all functions are clearly bounded by a constant if fixing $n$.

Now we can bound $R_n^1$ step by step \blue{using the expression of the posteriors given in Section \ref{sec:post}.} 
\blue{In the following subsections, we bound each of the three terms $V(s)$, $E_1(s)$ and $E_2(s)$.}

\subsubsection[step 1: bound V(s)]{Step 1: Bound $V(s)$}\label{sec:v}
We have $V(s) = M_{{Z}}(s, s) - M_{\hat{{Z}}}(s, s)$. It is easy to observe that 
\begin{align*}
	M_{{Z}}(s, s) &= \langle M_{{Z}}, K_s \otimes K_s \rangle = \left\langle \frac{1}{n} \sum_{i=1}^n (T_{{Z}} + \sigma^2_n)^{-1} \mu_{Z_i} \langle \mu_{Z_i}, K_s \rangle, K_s \right\rangle \\
	&= \left\langle (T_{{Z}} + \sigma^2_n)^{-1} T_{{Z}} K_s, K_s \right\rangle
\end{align*}
The second equation comes from the definition of $\mathcal{H} \otimes \mathcal{H}$. Therefore 
\begin{equation}\label{eq:var}
\begin{split}
	\mathbb E |V(s)| &= \mathbb E\left|\left\langle ((T_{{Z}} + \sigma^2_n)^{-1} T_{{Z}} - (T_{\hat{{Z}}} + \sigma^2_n)^{-1} T_{\hat{{Z}}}) K_s, K_s \right\rangle \right| \\
	&\le \kappa \mathbb E \norm{(T_{{Z}} + \sigma^2_n)^{-1} T_{{Z}} - (T_{\hat{{Z}}} + \sigma^2_n)^{-1} T_{\hat{{Z}}}}_{\mathcal{L(H)}} \\
	&= \kappa \sigma^2_n \mathbb E \norm{(T_{\hat{{Z}}} + \sigma^2_n)^{-1} - (T_{{Z}} + \sigma^2_n)^{-1}}_{\mathcal{L(H)}} \\
	&= \kappa \sigma^2_n \mathbb E \norm{(T_{\hat{{Z}}} + \sigma^2_n)^{-1} (T_{{Z}} - T_{\hat{{Z}}}) (T_{{Z}} + \sigma^2_n)^{-1}}_{\mathcal{L(H)}} \\
 &\le \frac{ \kappa}{\sigma_n^2} \mathbb E \norm{(T_{{Z}} - T_{\hat{{Z}}})}_{\mathcal{L(H)}} \\
 &\le \frac{2\sqrt{2}\kappa^2}{\sigma_n^2\sqrt{m}} 
 \end{split}
\end{equation}
\blue{from Lemma \ref{lem:tz} and applying Jensen's inequality with the concave function $\sqrt x$.}

\subsubsection[step 2: bound E1]{Step 2: Bound $E_1(s)$}\label{sec:e1}
One main difficulty for the bound here is that we don't assume $f_0$ lying in the RKHS of kernel $K$. In fact, when optimal rate is achieved with $\beta = \alpha$, the RKHS $\mathcal{H}_K$ contains all $(\beta + d/2)$-regular functions, which means $f_0 \not\in \mathcal{H}_K$. 

In the situation $f_0 \not\in \mathcal{H}_K$, from Lemma 4 of \cite{van2011information} we know that for order $\alpha$ Mat\'ern kernel $K$ and $\beta$-regular $f_0$ with $\beta \le \alpha$ we can always find $h \in \mathcal{H}_K$ such that:
\begin{equation} \label{f0_proj}
\inf_{\norm{h - f_0}_\infty < \varepsilon} \norm{h}_\mathcal{H}^2 \le C_\alpha \left(\frac{1}{\varepsilon}\right)^{(2\alpha - 2\beta + d)/\beta}
\end{equation}
for arbitrary small $\varepsilon$ and constant $C_\alpha \geq 1$ depending only on $\alpha$ and $f_0$.
Now we use equation \ref{f0_proj} to find a $f_0^\gamma \in \mathcal{H}$ such that $\norm{f_0^\gamma - f_0}_\infty \le n^{-\gamma}$ and $\norm{f_0^\gamma}_\mathcal{H} \le 2C_\alpha n^{(2\alpha-2\beta + d)\gamma/2\beta}$. We will determine $\gamma$ at the end of the proof. Then we have $\mathbb{E}_{\varepsilon_i} y_i = \mathbb{E}_{Z_i} f_0 = \mathbb{E}_{Z_i} f_0^\gamma + \mathbb{E}_{Z_i} (f_0 - f_0^\gamma) = \langle \mu_{Z_i}, f_0^\gamma \rangle + \mathbb{E}_{Z_i} (f_0 - f_0^\gamma)$. Denote $r_i = \mathbb{E}_{Z_i} (f_0 - f_0^\gamma)$, we have $|r_i| \le n^{-\gamma}$ and:
\begin{align*}
	\mathbb{E}_{\bm{\varepsilon}} (T_{{Z}} + \sigma_n^2)^{-1} g_{{Z}} &= (T_{{Z}} + \sigma_n^2)^{-1} \left(\frac{1}{n} \sum_i (\langle \mu_{Z_i}, f_0^\gamma \rangle \mu_{Z_i} + r_i \mu_{Z_i})\right) \\
	&= (T_{{Z}} + \sigma_n^2)^{-1} T_{{Z}} f_0^\gamma + \frac{1}{n}\sum_i r_i (T_{{Z}} + \sigma_n^2)^{-1} \mu_{Z_i}
\end{align*}
And similarly:
\begin{align*}
	\mathbb{E}_{\bm{\varepsilon}} (T_{\hat{{Z}}} + \sigma_n^2)^{-1} g_{\hat{{Z}}} &= (T_{\hat{{Z}}} + \sigma_n^2)^{-1} \left(\frac{1}{n} \sum_i (\langle \mu_{\widehat Z_i}, f_0^\gamma \rangle \mu_{\widehat{Z}_i} + r_i \mu_{\widehat{Z}_i})\right) \\
	&= (T_{\hat{{Z}}} + \sigma_n^2)^{-1} T_{\hat{{Z}}} f_0^\gamma + \frac{1}{n}\sum_i (r_i + d_i) (T_{\hat{{Z}}} + \sigma_n^2)^{-1} \mu_{\widehat{Z}_i}
\end{align*}
where $d_i = \langle \mu_{Z_i} - \mu_{\widehat{Z}_i}, f_0^\gamma \rangle$. 

Therefore
\begin{align*}
	|E_1(s)| &=  \left|\langle E_1, K_s \rangle\right| \le \sqrt{\kappa} \norm{E_1}_\mathcal{H} \\
	&\le \sqrt{\kappa} \mathbb E \norm{(T_{{Z}} + \sigma_n^2)^{-1} T_{{Z}} - (T_{\hat{{Z}}} + \sigma_n^2)^{-1} T_{\hat{{Z}}}}_{\mathcal{L(H)}} \norm{f_0^\gamma}_\mathcal{H} \\
	&\hspace{2em}+ \sqrt{\kappa} \norm{\frac{1}{n}\sum_i d_i (T_{\hat{{Z}}} + \sigma_n^2)^{-1} \mu_{\widehat{Z}_i}}_\mathcal{H} + \sqrt{\kappa} \norm{\frac{1}{n}\sum_i r_i (T_{{Z}} + \sigma_n^2)^{-1} (\mu_{Z_i} - \mu_{\widehat{Z}_i})}_\mathcal{H} \\
	&\hspace{2em}+ \sqrt{\kappa} \norm{\frac{1}{n}\sum_i r_i ((T_{{Z}} + \sigma_n^2)^{-1} - (T_{\hat{{Z}}} + \sigma_n^2)^{-1}) \mu_{\widehat{Z}_i}}_\mathcal{H} \\ 
	&\leq \left(\frac{2 C_\alpha \sqrt{\kappa}}{\sigma_n^2} \norm{T_{{Z}} - T_{\hat{{Z}}}}_{\mathcal{L(H)}} + \frac{2C_\alpha\kappa}{\sigma_n^2 n}\sum_i \norm{\mu_{Z_i} - \mu_{\widehat{Z_i}}}_\mathcal{H} \right) n^{\frac{(2\alpha-2\beta + d) \gamma}{2\beta}} \\
	&\hspace{2em}+ \left(\frac{\sqrt{\kappa}}{\sigma_n^2 n} \sum_i \norm{\mu_{Z_i} - \mu_{\widehat{Z_i}}}_\mathcal{H} + \frac{\kappa}{\sigma_n^4} \norm{T_{{Z}} - T_{\hat{{Z}}}}_{\mathcal{L(H)}}\right) n^{-\gamma}
\end{align*}

Using the bounds \blue{from Lemmas \ref{lem:kme} and \ref{lem:tz},} we have:
\begin{align*}
	\left| \mathbb{E} \int_{[0,1]^d} 2f_0(s) E_1(s) ds\right| \le \frac{16C_0 C_\alpha \kappa^{3/2}}{\sigma_n^2 \sqrt{m}} n^{\frac{(2\alpha-2\beta + d) \gamma}{2\beta}} + \frac{4C_0\kappa}{\sigma_n^2 \sqrt{m}} n^{-\gamma} + \frac{4\sqrt{2} \kappa^2}{\sigma_n^4 \sqrt{m}} n^{-\gamma}
\end{align*}
where $C_0 \geq 1$ is the upper bound of $|f_0|$ which must exists, because we assume $f_0$ regular (hence continuous) in a compact set.
The rate can be improved when $f_0 \in \mathcal{H}_K$ ($\beta \ge \alpha + d/2$), in which case $f_0^\gamma$ can directly be $f_0$, making $r_i = 0$ and $\norm{f_0^\gamma}$ constant. We have:
\begin{align*}
	\left| \mathbb{E} \int_{[0,1]^d} 2f_0(s) E_1(s) ds\right| \le \frac{16C_0^2 C_\alpha \kappa^{3/2}}{\sigma_n^2 \sqrt{m}}
\end{align*}

\subsubsection[step 3: bound E2]{Step 3: Bound $E_2(s)$}\label{sec:e2}

Using the same notation as in Step 2. In the situation $f_0 \not\in \mathcal{H}_K$, we denote:
\begin{align*}
	A &= (T_{{Z}} + \sigma_n^2)^{-1}T_{{Z}} f_0^\gamma, \hspace{1em} B = \frac{1}{n} \sum_i \varepsilon_i (T_{{Z}} + \sigma_n^2)^{-1} \mu_{Z_i}, \hspace{1em} C = \frac{1}{n} \sum_i r_i (T_{{Z}} + \sigma_n^2)^{-1} \mu_{Z_i} \\
	\hat{A} &= (T_{\hat{{Z}}} + \sigma_n^2)^{-1}T_{\hat{{Z}}} f_0^\gamma, \hspace{1em}\hat{B} = \frac{1}{n} \sum_i \varepsilon_i (T_{\hat{{Z}}} + \sigma_n^2)^{-1} \mu_{\widehat{Z}_i} , \hspace{1em} \hat{C} = \frac{1}{n} \sum_i (r_i + d_i) (T_{\hat{{Z}}} + \sigma_n^2)^{-1} \mu_{\widehat{Z}_i}
\end{align*}
Then $E_2 = \mathbb{E}_{\bm{\varepsilon}} \left[(A+B+C)^2 - (\hat{A} + \hat{B} + \hat{C})^2\right]$. Easy to see $\mathbb{E}_{\bm{\varepsilon}} \left[B(A+C) + \hat{B}(\hat{A} + \hat{C})\right] = 0$. Because $\varepsilon_i$ is mean 0 and independent of any $Z_i$ and $\bx_{ij}$. 
Therefore, 
\begin{equation*}
	E_2 = (A+C+\hat{A}+\hat{C})(A-\hat{A}+C-\hat{C}) + \mathbb{E}_{\bm{\varepsilon}} \left[B^2 - \hat{B}^2\right]
\end{equation*}
For the first part
\begin{align*}
	&|(A+C+\hat{A}+\hat{C})(A-\hat{A}+C-\hat{C})|(s) = |\langle A+C+\hat{A}+\hat{C}, K_s\rangle \langle A-\hat{A}+C-\hat{C}, K_s\rangle| \\
	&\hspace{6em}\le \kappa \left(\norm{A}_\mathcal{H} + \norm{\hat{A}}_\mathcal{H} + \norm{C}_\mathcal{H} + \norm{\hat{C}}_\mathcal{H}\right) \norm{E_1}_\mathcal{H} \\
	&\hspace{6em}\le \kappa \norm{E_1}_\mathcal{H} \left(2\norm{f_0^\gamma}_\mathcal{H} + \frac{2\sqrt{\kappa}}{\sigma_n^2} n^{-\gamma} + \frac{\sqrt{\kappa}\norm{f_0^\gamma}_\mathcal{H}}{\sigma_n^2 n}\sum_i \norm{\mu_{Z_i} - \mu_{\widehat{Z}_i}}_\mathcal{H} \right)
\end{align*}
We can use the same bound for $E_1$ and \blue{apply Lemmas \ref{lem:kme} and \ref{lem:tz} to have:}
\begin{align*}
	&\mathbb{E} \left[\frac{\norm{T_{{Z}}-T_{\hat{{Z}}}}_{\mathcal{L(H)}}}{n}\sum_i \norm{\mu_{Z_i} - \mu_{\widehat{Z}_i}}_\mathcal{H}\right] \le \sqrt{\mathbb{E}\norm{T_{{Z}}-T_{\hat{{Z}}}}_{\mathcal{L(H)}}^2 \mathbb{E}\norm{\mu_{Z_i} - \mu_{\widehat{Z}_i}}_\mathcal{H}^2} = O\left(\frac{1}{m}\right)\\
	&\mathbb{E} \left[\left(\frac{1}{n} \sum_i \norm{\mu_{Z_i} - \mu_{\widehat{Z}_i}}_\mathcal{H}\right) \left(\frac{1}{n} \sum_i \norm{\mu_{Z_i} - \mu_{\widehat{Z}_i}}_\mathcal{H} \right)\right] \le \sqrt{\mathbb{E}\norm{\mu_{Z_i} - \mu_{\widehat{Z}_i}}_\mathcal{H}^2 \mathbb{E}\norm{\mu_{Z_i} - \mu_{\widehat{Z}_i}}_\mathcal{H}^2} = O\left(\frac{1}{m}\right)
\end{align*}

For the second part, we have (notice that $(T_{{Z}} + \sigma_n^2)^{-1}$ is Hermitian):
\begin{align*}
	\mathbb{E}_{\bm{\varepsilon}} [B^2] (s) &= \frac{1}{n^2} \sum_i \sigma^2 \langle(T_{{Z}} + \sigma_n^2)^{-1} \mu_{Z_i}, K_s\rangle \langle \mu_{Z_i}, (T_{{Z}} + \sigma_n^2)^{-1} K_s\rangle \\
	&= \frac{\sigma^2}{n} \langle(T_{{Z}} + \sigma_n^2)^{-2} T_{{Z}} K_s, K_s \rangle
\end{align*}
and  
\begin{align*}
	|\mathbb{E}_{\bm{\varepsilon}} [B^2-\hat{B}^2] (s)| &\le \frac{\sigma^2 \kappa}{n} \norm{(T_{{Z}} + \sigma_n^2)^{-2} T_{{Z}} - (T_{\hat{{Z}}} + \sigma_n^2)^{-2} T_{\hat{{Z}}}}_{\mathcal{L(H)}} \\
	&\le \kappa\sigma_n^2 \norm{((T_{{Z}} + \sigma_n^2)^{-1} - (T_{\hat{{Z}}} + \sigma_n^2)^{-1})(T_{{Z}} + \sigma_n^2)^{-1}T_{{Z}}}_{\mathcal{L(H)}} \\
	&\hspace{1em}+ \kappa\sigma_n^2\norm{(T_{{Z}} + \sigma_n^2)^{-1}((T_{{Z}} + \sigma_n^2)^{-1}T_{{Z}} - (T_{\hat{{Z}}} + \sigma_n^2)^{-1}T_{\hat{{Z}}})}_{\mathcal{L(H)}} \\
	&\le \frac{2\kappa}{\sigma_n^2} \norm{T_{{Z}} - T_{\hat{{Z}}}}_{\mathcal{L(H)}}
\end{align*}

Therefore, using bound \blue{from Lemma \ref{lem:tz}} and ignoring the constant term, we have (recall that $\norm{f_0^{\gamma}}_\mathcal{H} = O(n^{(2\alpha-2\beta+d)\gamma / 2\beta}) $):
\begin{align} \label{rn1_bound}
	\left|\mathbb{E} \int_{[0,1]^d} E_2(s) ds\right| &\le C(\alpha,\kappa,f_0) \left( \frac{1}{\sigma_n^2\sqrt{m}} n^{\frac{(2\alpha-2\beta+d)\gamma}{\beta}} + \frac{1}{\sigma_n^4\sqrt{m}}n^{-\gamma + \frac{(2\alpha-2\beta+d)\gamma}{2\beta}}\right. \\
	&\hspace{1em}+ \left. \frac{1}{\sigma^6_n \sqrt{m}} n^{-2\gamma} + \frac{1}{\sigma_n^4 m} n^{\frac{(2\alpha-2\beta+d)\gamma}{\beta}} + \frac{1}{\sigma_n^6 m} n^{-\gamma + \frac{(2\alpha-2\beta+d)\gamma}{2\beta}}\right)
\end{align}
for some constant $C(\alpha,\kappa,f_0)$ depending only on $\alpha, \kappa, f_0$.

Similarly, the rate can be improved if $f_0 \in \mathcal{H}_K$, where $f_0^\gamma$ can directly be $f_0$, making $r_i = 0$ and $\norm{f_0^\gamma}$ constant. We have:
\begin{align} \label{rn1_bound2}
	\left|\mathbb{E} \int_{[0,1]^d} E_2(s) ds\right| &\le \frac{C(\kappa,f_0)}{\sigma^2_n \sqrt{m}}
\end{align}
given $\sigma^2_n \sqrt{m} \to \infty$. Notice again that $\sigma_n^2 = \sigma^2 / n$.

\subsubsection{Combine Together}\label{sec:combine}
\blue{Combining the three bounds from Sections \ref{sec:v} to \ref{sec:e2}, we have established (\ref{eq:r1rate_in}) 
if $f_0 \in \mathcal{H}_K$ ($\beta \ge \alpha + d/2$) and (\ref{eq:r1rate_out}) if $f_0 \not\in \mathcal{H}_K$.}

\subsection[Bound R1 under fixed design ]{Bound $R_n^1$ \blue{under fixed design  (Theorem \ref{th:fixed})}}
Denote $\hat{\mathcal{Z}}$ as the empirical distributional process supported only on $n$ points of $Z_i$. We have:
\begin{align*}
	R_n^1 &= \mathbb{E}_{f_0} \int \norm{f-f_0}_n^2 (d\Pi_n(f |\mathbb{D}_n) - d\Pi_n(f |\mathbb{Z}_n)) \\
	& = \mathbb{E}_{f_0} \int \mathbb{E}_{z \sim \hat{\mathcal{Z}}} [\mathbb{E}_z(f-f_0) \mathbb{E}_z(f-f_0)] (d\Pi_n(f |\mathbb{D}_n) - d\Pi_n(f |\mathbb{Z}_n)) \\
	& = \mathbb{E}_{f_0} \mathbb{E}_{z \sim \hat{\mathcal{Z}}} \int \left[\int (f-f_0)(s)(f-f_0)(t) dz(s) dz(t)\right] (d\Pi_n(f |\mathbb{D}_n) - d\Pi_n(f |\mathbb{Z}_n))\\
	& = \mathbb{E}_{f_0} \mathbb{E}_{z \sim \hat{\mathcal{Z}}} \int (\mathbb{E}_{f|\mathbb{D}_n} - \mathbb{E}_{f|\mathbb{Z}_n}) [f(s)f(t)] - 2f_0(s) (\mathbb{E}_{f|\mathbb{D}_n} - \mathbb{E}_{f|\mathbb{Z}_n}) [f(t)]\  dz(s) dz(t)
\end{align*}
It is not hard to see one can use exactly the same method as in section \ref{sec:morder} to get the same bound as for the term $|(\mathbb{E}_{f|\mathbb{D}_n} - \mathbb{E}_{f|\mathbb{Z}_n}) [f(s)f(t)] - 2f_0(s) (\mathbb{E}_{f|\mathbb{D}_n} - \mathbb{E}_{f|\mathbb{Z}_n}) [f(t)]|$ with every $s,t$ pair, and therefore we can get the same overall bounds as for $\norm{\cdot}_2$. 

\subsection[ Bound R0]{Bound for $R^0_n$}\label{sec:banach}
We would use the method described in \cite{van2011information} to bound $R^0_n$, by extending it to non-parametric GP regression on the space of distributions. 
\blue{As described in Section \ref{sec:r0}, we write (\ref{eq:gpan2}) as a direct GP regression on distributional inputs \eqref{eq:gpan2dist} with the GP taking value in a linear space $\calB$.}
Naturally we give $\mathcal{B}$ a norm that $\norm{F}_\mathcal{B} = \norm{\pi(F)}_\infty$ and it is clear that $|F(Z)| \le \norm{F}_\mathcal{B}$ for all $Z$. When it is not misleading, we will also denote $\pi(F)$ simply as $f$. Similarly, $f_0$ for $\pi(F_0)$ and any other super-sub-script. 

It is clear that, if $\mathcal{Z}$ weakly separates $\mathcal{C}_0$, $\pi$ would be an isomorphism between $\mathcal{B}$ and $\mathcal{C}_0$. Because $\pi$ would be a bijection with $\norm{F}_\mathcal{B} = \norm{\pi(F)}_\infty$. Hence, $\mathcal{B}$ is a separable Banach space.

Denote the RKHS of $\mathbb{K}$ to be $\mathcal{H}_\mathbb{K}$. Then $\mathcal{H}_\mathbb{K}$ is a subspace of $\mathcal{B}$ when using Mat\'ern kernel for $K$. From Lemma \ref{lem:hk}, we also have $\pi(\mathcal{H}_\mathbb{K}) = \mathcal{H}$, with $\norm{h}_{\mathcal{H}_\mathbb{K}} = \norm{\pi(h)}_\mathcal{H}$. 

\begin{lemma} \label{lem:hk}
	Using the kernel $\mathbb{K}$ and projection $\pi$ defined above, and $\mathcal{H}_\mathbb{K}$ be the RKHS from $\mathbb{K}$. We have $\pi(\mathcal{H}_\mathbb{K}) = \mathcal{H}$ with $\norm{h}_{\mathcal{H}_\mathbb{K}} = \norm{\pi(h)}_\mathcal{H}$ if $\mathcal{Z}$ weakly separates $\mathcal{C}_0$.
\end{lemma}

\begin{proof}
	First, $\mathcal{H}_\mathbb{K}$ is the Hilbert space spanned by $\mathbb{K}(Z, \cdot) = \langle \mu_Z, \mu_{(\cdot)} \rangle_\mathcal{H}$ and it is easy to see $\pi(\mathbb{K}(Z, \cdot)) = \mu_Z = \int K(s, \cdot) dZ(s)$. Because $\langle \mu_w, g \rangle_{\mathcal{H}} = \mathbb{E}_w f$ for any distribution $w$ and function $g \in \mathcal{H}$. Also $\norm{\mathbb{K}(Z, \cdot)}_{\mathcal{H}_\mathbb{K}} = \mathbb{K}(Z, Z) = \int K(s,t)dZ \otimes dZ = \norm{\pi(\mathbb{K}(Z, \cdot))}_\mathcal{H}$. Since $\pi$ is an isomorphism between $\mathcal{B}$ and $\mathcal{C}_0$. We know that $\pi(\mathcal{H}_\mathbb{K})$ is a subspace in $\mathcal{H}$ spanned by $\{\mu_Z:\ Z\in \mathcal{S}\}$ with $\norm{\pi(h)}_\mathcal{H} = \norm{h}_{\mathcal{H}_\mathbb{K}}$. 
	
	Now decompose $\mathcal{H} = \pi(\mathcal{H}_\mathbb{K}) \oplus \pi(\mathcal{H}_\mathbb{K})^\perp$, for any $f \in \pi(\mathcal{H}_\mathbb{K})^\perp$ we have:
	$$\forall Z\in \mathcal{S}:\ \langle \mu_Z, f \rangle_\mathcal{H} = \mathbb{E}_Z f = 0$$
	but by our assumption on the richness of $\mathcal{S}$, this can happen only when $f = 0$. Therefore, $\pi(\mathcal{H}_\mathbb{K})^\perp = \bm{0}$ hence $\mathcal{H} = \pi(\mathcal{H}_\mathbb{K})$.
\end{proof}

Our model \ref{eq:gpan2} with known individual distributions is equivalent to Gaussian process regression:
\begin{align}
	&F \sim GP(0, \mathbb{K}) \\
	&y_i \sim \mathcal{N}(F(Z_i), \sigma^2)
\end{align}

It is not hard to see that $f \sim GP(0, K) \Rightarrow \pi^{-1}(f) \sim GP(0, \mathbb{K})$. And by definition $\mathbb{E}_{Z_i} f = \pi^{-1}(f)(Z_i)$ for all $Z_i$. Therefore, from the uniqueness of Gaussian process we know $\pi^{-1}(f) | \mathbb{Z}_n \sim F | \mathbb{Z}_n$.

Now, define square operator $\mathcal{B} \to \mathcal{B}$ as $F^2 = \pi^{-1}(\pi(F)^2)$, where square in $\mathcal{C}_0$ is the typical point-wise square. Since the square of a bounded continuous function is still bounded continuous, this square operator in $\mathcal{B}$ is well-defined. And obviously $F^2(Z) = \mathbb{E}_Z \pi(F)^2$. Now define $L_2$ norm in $\mathcal{B}$ as $\norm{F}^2_2 = F^2(\mu)$, where $\mu$ is the mean measure of $Z_i$ (assumed to be Unif(0, 1)). It is clear that $\norm{F}_2 \le \norm{F}_\mathcal{B}$. Similarly, define empirical norm in $\mathcal{B}$ as $\norm{F}_n^2 = \frac{1}{n}\sum_i^n (F(Z_i))^2$. It is not had to see risk $R_n^0$ agrees for both models under $\norm{\cdot}_2$ and $\norm{\cdot}_n$. For example, for $\norm{\cdot}_2$:
\begin{align*}
	R_n^0 &= \mathbb{E}_{f_0} \int \norm{f-f_0}_2^2 d \Pi_n(f|\mathbb{Z}_n) = \mathbb{E}_{f_0} \int \mathbb{E}_{\mu} (f-f_0)^2 d \Pi_n(f|\mathbb{Z}_n) \\
	&= \mathbb{E}_{f_0} \int \norm{\pi^{-1}(f) - F_0}_2^2 d \Pi_n(f|\mathbb{Z}_n) = \mathbb{E}_{f_0} \int \norm{F - F_0}_2^2 d \Pi_n(F|\mathbb{Z}_n)
\end{align*}
\blue{and similarly for $\norm{\cdot}_n$. We now establish the bounds for fixed design and random design.}

\noindent \blue{\textbf{Fixed design (Theorem \ref{th:fixed}):}} Using Theorem 1 in \cite{van2011information} we can bound the risk term $R_n^0$ in fixed design situation with $\Psi^{-1}_{F_0}(n)^2$, where $F\sim GP(0, \mathbb{K})$ and:
\begin{equation} \label{eq:control}
\varepsilon^2 \Psi_{F_0}(\varepsilon) = \inf_{h \in \mathcal{H}_\mathbb{K}: ||h-F_0||_\mathcal{B} < \varepsilon} ||h||^2_{\mathcal{H}_\mathbb{K}} - \log \mathbb{P}(||F||_\mathcal{B} < \varepsilon)
\end{equation}
Recall the definition of $\mathcal{H}_\mathbb{K}$ and $\mathcal{B}$, and consider $f \sim GP(0, K)$ we have:
\begin{align*}
	\mathbb{P}(\norm{f}_\infty \le \varepsilon) &= \mathbb{P}\left(\norm{\pi^{-1}(f)}_\mathcal{B} \le \varepsilon\right) = \mathbb{P}(||F||_\mathcal{B} < \varepsilon) \\
	\inf_{h \in \mathcal{H}: ||h-f_0||_\infty < \varepsilon} ||h||^2_{\mathcal{H}} &= \inf_{h \in \mathcal{H}: ||\pi^{-1}(h)-F_0||_\mathcal{B} < \varepsilon} ||\pi^{-1}(h)||^2_{\mathcal{H}_\mathbb{K}} = \inf_{h \in \mathcal{H}_\mathbb{K}: ||h-F_0||_\mathcal{B} < \varepsilon} ||h||^2_{\mathcal{H}_\mathbb{K}}
\end{align*}
Therefore $\Psi_{F_0}(\varepsilon) = \Psi_{f_0}(\varepsilon)$ with:
\begin{equation} \label{psi_f0}
\varepsilon^2 \Psi_{f_0}(\varepsilon) = \inf_{h \in \mathcal{H}: ||h-f_0||_\infty < \varepsilon} ||h||^2_{\mathcal{H}} - \log \mathbb{P}(||f||_\infty < \varepsilon)
\end{equation}
which means we can get the same rates as in Theorem 5 of \cite{van2011information}, which is $n^{-2\min(\alpha, \beta)/(2\alpha + d)}$.

\noindent \blue{\textbf{Random design (Theorem \ref{th:random}):}} 
 Strong separation is needed for the estimation bound because empirical norm $\norm{f}^2_n = \frac{1}{n}\sum_i (\mathbb{E}_{Z_i} f)^2$ converges to $\mathbb{E}_{Z\sim \mathcal{Z}}[(\mathbb{E}_Z f)^2]$, not to $\norm{f}^2_2$. Introducing a gap between the empirical bound and estimation bound. But if $\mathcal{Z}$ strongly separates $\pi(\mathcal{B})$ with constant $C$. Then $\norm{f}^2_2$ becomes equivalent norm with \blue{$\|f\|_{\mathcal Z}^2 := \mathbb{E}_{Z\sim \mathcal{Z}}[(\mathbb{E}_Z f)^2]$} and the proof for Theorem 2 in \cite{van2011information} can be continued by observing:
\begin{align*}
	&P(\norm{f-f_\varepsilon}_2 \ge 2C\norm{f-f_\varepsilon}_n) \le P(\norm{f-f_\varepsilon}_\mathcal{Z} \ge 2\norm{f-f_\varepsilon}_n) \\
	&\hspace{3em} \le e^{-(n/5) \norm{f - f_\varepsilon}_\mathcal{Z}^2 / \norm{f - f_\varepsilon}_\infty^2} \le e^{-(n/5C^2) \norm{f - f_\varepsilon}_2^2 / \norm{f - f_\varepsilon}_\infty^2}
\end{align*}
for any $f$ and $f_\varepsilon$. In such case we can get the same rate as $n^{-2\min(\alpha, \beta)/(2\alpha + d)}$, given $\min(\alpha, \beta) > d / 2$.

\subsection{\blue{Proof for dependent samples (Theorem \ref{th:dep})}}\label{sec:pfdep}
\begin{proof}[\blue{Proof of Theorem \ref{th:dep}}] \blue{We note that in the proof of the theorems for the i.i.d. case the samples $\bx_{ij}$ only play a part in the excess risk term $R_n^0$ of (\ref{eq:decomp}), and the error bounds of $R_n^0$ depend on $\bx_{ij}$ exclusively through the result in Lemma \ref{lem:kme} on the distances between the mean embeddings between the true and observed distributions. So it is enough to show that the bound in Lemma \ref{lem:kme} for this dependent setup.

Following the decomposition of $\norm{\mu_{\widehat{Z}_i} - \mu_{Z_i}}_\mathcal{H}^2$ provided in (\ref{eq:cross}), the square terms $\frac 1{m^2} \sum_{j=1}^m \norm{K_{\bx_{ij}} - \mu_{Z_i}}_\mathcal{H}^2$ still has the same distribution as in the i.i.d. case as $\bx_{ij}$ are identically distributed and thus its expectation has the same $O(1/m)$ bound derived in (\ref{eq:kmebound}). 

In the i.i.d case the cross-terms $\frac 1{m^2} \sum_{j \neq j'} <K_{\bx_{ij}} - \mu_{Z_i},K_{\bx_{ij'}} - \mu_{Z_i}>_{\mathcal H}$ had $0$ expectation and this will not be the case  under assumed dependence of the samples. Instead we will show that the expectation of the sum of the cross-terms will also be $O(1/m)$ under the restriction on the $\beta$-mixing coefficients (uniform geometric decay). 

Consider an infinite sequence of sample $\bx_{i1}, \bx_{i2}, \ldots$ that are identically distributed draws from $Z_i$ with geometrically $\beta$-mixing coefficients $\beta_i(k)$. Let $Z_{i,jj'}$ denote the joint distribution of $\bx_{ij}$ and $\bx_{ij'}$.
Then we have:
\begin{equation}\label{eq:crossdep}
\begin{split}
   & \mathbb E_{\bx_{ij},\bx_{ij'}}  |<K_{\bx_{ij}} - \mu_{Z_i},K_{\bx_{ij'}} - \mu_{Z_i}>_{\mathcal H}| \\
    &\quad = | \int_{\bx} \int_{\bx'} K(\bx,\bx') Z_{i,jj'}(d(\bx\bx')) - \int_{\bx} \int_{\bx'} K(\bx,\bx')  Z_i(d\bx) Z_i(d\bx')|.
\end{split}
\end{equation}

The Mat\'ern kernel $K$ is a continuous positive function on $\calX \times \calX$, and thus measurable with respect to the Borel $\sigma$-algebra $\mathbb B(\calX \times \calX)$. It is also bounded by $\kappa$. So it can be approximated by bounded simple (step) functions $K_r \uparrow K$
pointwise. We can write $K_r(\bx,\bx') = \sum_{k=1}^{L_r} \sum_{k'=1}^{L'_r} c_{kk',r} I(\bx \in A_{k,r})I(\bx' \in B_{k',r})$ where $A_{k,r}, B_{k',r} \in \mathbb B(\calX)$ for all $k,k',r$, $\{A_{k,r}\}_{k=1,2,\ldots,L_r}$ denotes a partition of $\calX$, $\{B_{k',r}\}_{k'=1,2,\ldots,L'_r}$ denotes another partition of $\calX$, and $0 \leq c_{kk',r} \leq \kappa$. There are two levels of approximation involved in going from $K$ to $K_r$, first the approximation of $K$ in terms of increasing simple functions on a partition of $\mathbb B(\calX \times \calX)$, then approximating a Borel set in $\mathbb B(\calX \times \calX)$ as union of rectangles $A_{k,r} \times B_{k',r}$ in $\mathbb B(\calX) \times \mathbb B(\calX)$.

Then  we have 
\begin{equation*}
\begin{split}
    & | \int_{\bx} \int_{\bx'} K_r(\bx,\bx') Z_{i,jj'}(d(\bx\bx')) - \int_{\bx} \int_{\bx'} K_r(\bx,\bx')  Z_i(d\bx) Z_i(d\bx')|\\
    & \quad \leq \sum_{k} c_{kk',r} |P(\bx_{ij} \in A_{k,r},\bx_{ij'} \in B_{k',r}) - P(\bx_{ij} \in A_{k,r})P(\bx_{ij'} \in B_{k',r})| \\
    & \quad \leq \kappa \beta_i(|j-j'|)
\end{split}
\end{equation*}
from the definition of $\beta$-mixing coefficients (see, e.g., Eq (1.5) of \cite{bradley}). As $K_r \to K$ using monotone convergence theorem we have $$\int_{\bx} \int_{\bx'} K_r(\bx,\bx') Z_{i,jj'}(d(\bx\bx')) \to \int_{\bx} \int_{\bx'} K(\bx,\bx') Z_{i,jj'}(d(\bx\bx')), \mbox{ and }$$
$$  \int_{\bx} \int_{\bx'} K_r(\bx,\bx')  Z_i(d\bx) Z_i(d\bx') \to  \int_{\bx} \int_{\bx'} K(\bx,\bx')  Z_i(d\bx) Z_i(d\bx').$$
Returning to (\ref{eq:crossdep}), we then have 
\begin{align*}
    \mathbb E_{\bx_{ij},\bx_{ij'}}  |<K_{\bx_{ij}} - \mu_{Z_i},K_{\bx_{ij'}} - \mu_{Z_i}>_{\mathcal H}| \leq \kappa \beta_i(|j - j'|).
\end{align*}
Consequently, letting $\mathbb N$ denote the set of natural numbers and $\psi = \sup_i \sum_{k=1} \beta_i(k)$ we have 
\begin{align*}
    & \frac 1{m^2} \sum_{j \neq j'} \mathbb E_{\bx_{ij},\bx_{ij'}}  |<K_{\bx_{ij}} - \mu_{Z_i},K_{\bx_{ij'}} - \mu_{Z_i}>_{\mathcal H}| \\
    & \quad \leq \frac \kappa {m^2}\sum_{j=1}^m \sum_{j' \in \mathbb N \setminus \{j\}} 
    \beta_i(|j - j'|)\\
    & \quad \leq \frac {\kappa} {m^2}\sum_{j=1}^m \psi \\
    & \quad \leq \frac {\psi \kappa} m.
\end{align*}

}


\end{proof}

\end{document}